\documentclass[pra,aps,reprint,a4paper,superscriptaddress,floatfix]{revtex4-2}
\usepackage{times}
\usepackage{graphicx}
\usepackage{epsfig}
\usepackage{subfigure}
\usepackage{amsfonts}
\usepackage{amsmath}
\usepackage{amssymb}
\usepackage{dsfont}
\usepackage{amsthm}
\usepackage{color}
\usepackage[colorlinks=true,linkcolor=blue,citecolor=blue,urlcolor=blue]{hyperref}
\usepackage{comment}
\usepackage{braket}
\usepackage{physics}
\usepackage{multirow}
\usepackage{float}

\usepackage[latin9]{inputenc}
\usepackage{amsmath}
\usepackage{amssymb, enumerate}
\usepackage{color}
\usepackage{babel}
\usepackage{graphicx}
\usepackage{epstopdf}
\usepackage{float}
\usepackage{thmtools}
\usepackage{thm-restate}
\usepackage{hyperref}
\usepackage{cleveref}
\usepackage{enumerate}
\usepackage{lipsum}
\usepackage{fixmath}
\usepackage{dsfont}

\usepackage{amsthm}
\usepackage{amsmath}
\usepackage{amssymb}
\usepackage{tikz}
\usepackage{color}
\usetikzlibrary{matrix}
\usepackage{dsfont}
\usepackage{graphicx} 
\usepackage{stackrel}
\makeatletter

\@ifundefined{textcolor}{}
{%
 \definecolor{BLACK}{gray}{0}
 \definecolor{WHITE}{gray}{1}
 \definecolor{RED}{rgb}{1,0,0}
 \definecolor{GREEN}{rgb}{0,1,0}
 \definecolor{BLUE}{rgb}{0,0,1}
 \definecolor{CYAN}{cmyk}{1,0,0,0}
 \definecolor{MAGENTA}{cmyk}{0,1,0,0}
 \definecolor{YELLOW}{cmyk}{0,0,1,0}
}

\newenvironment{protocol*}[1]
  {
    \begin{center}
      \hrulefill\\
      \textbf{#1}
  }
  {
    \vspace{-1\baselineskip}
    \hrulefill
    \end{center}
  }

\newtheorem{thm}{Theorem}

\newtheorem{pro}[thm]{Proposition}

\theoremstyle{definition}

\newtheorem{lemma}{Lemma}  
\def\bel{\begin{lemma}}
\def\eel{\end{lemma}}
\newtheorem{theorem}{Theorem}

\newtheorem*{proposition*}{Proposition}
\newtheorem{proposition}[theorem]{Proposition}

\newtheorem{definition}[theorem]{Definition}

\makeatother

\begin{document}

\title{Generalised Kochen-Specker Theorem for Finite Non-Deterministic Outcome Assignments}
\author{Ravishankar Ramanathan}
\email{ravi@cs.hku.hk}
\affiliation{Department of Computer Science, The University of Hong Kong, Pokfulam Road, Hong Kong}


\begin{abstract}
The Kochen-Specker (KS) theorem is a cornerstone result in quantum foundations, establishing that quantum correlations in Hilbert spaces of dimension $d \geq 3$ cannot be explained by (consistent) hidden variable theories that assign a single deterministic outcome to each measurement. Specifically, there exist finite sets of vectors in these dimensions such that no non-contextual deterministic ($\{0,1\}$) outcome assignment is possible obeying the rules of exclusivity and completeness - that the sum of assignments to every set of mutually orthogonal vectors be $\leq 1$ and the sum of value assignments to any $d$ mutually orthogonal vectors be equal to $1$. Another central result in quantum foundations is Gleason's theorem that justifies the Born rule as a mathematical consequence of the quantum formalism. The KS theorem can be seen as a consequence of Gleason's theorem and the logical compactness theorem. In a similar vein, Gleason's theorem also indicates the existence of KS-type finite vector constructions to rule out other finite alphabet outcome assignments beyond the $\{0,1\}$ case. Here, we propose a generalisation of the KS theorem that rules out hidden variable theories with outcome assignments in the set $\{0, p, 1-p, 1\}$ for $p \in [0,1/d) \cup (1/d, 1/2]$. The case $p = 1/2$ is especially physically significant. We show that in this case the result rules out (consistent) hidden variable theories that are fundamentally binary, i.e., theories where each measurement has fundamentally at most two outcomes (in contrast to the single deterministic outcome per measurement ruled out by KS). We present a device-independent application of this generalised KS theorem by constructing a two-player non-local game for which a perfect quantum winning strategy exists (a Pseudo-telepathy game) while no perfect classical strategy exists even if the players are provided with additional no-signaling resources of PR-box type (with marginals in $\{0,1/2,1\}$).

\end{abstract}



\maketitle

\textit{Introduction.-}
One of the most striking features of the quantum world is Contextuality \cite{KS67, Bell66, Held16, BCGKL22}, the notion that outcomes cannot be assigned to measurements independently of the particular contexts in which the measurements are realized. Simply put, the measurements of quantum observables may not be thought of as revealing pre-existing properties that are independent of other compatible observables measured on the system. This feature has in recent years gone beyond being a fundamental curiosity and has been established as a fundamental resource in quantum information processing tasks including magic state distillation \cite{HWVE14}, measurement-based quantum computation \cite{Raus13}, semi-device-independent randomness generation \cite{ACCS12, U+13, H+10, our18, ZRLH23}, zero-error information theory \cite{CLMW10}, non-local pseudo-telepathy games \cite{RW04}, as well as most recently in communication complexity scenarios \cite{GSXCM23}.

The phenomenon of contextuality (specifically outcome contextuality) manifests itself in the form of the famous Kochen-Specker (KS) theorem \cite{KS67}. The KS theorem states that for every quantum system belonging to a Hilbert space of dimension greater than two, irrespective of its actual state, a finite set of measurements exists that does not admit a deterministic non-contextual outcome assignment. Here, deterministic indicates that the outcome probabilities only take values in $\{0,1\}$. Non-contextual means that the assignment is made to the individual projectors, independently of the context (the other projectors in the measurements) to which they belong. The KS theorem is also formulated as saying that in Hilbert spaces of dimension greater than two, there exist finite sets of vectors that are not $\{0,1\}$-colorable, wherein a $\{0,1\}$-coloring is such a deterministic non-contextual outcome assignment obeying the natural Kochen-Specker rules that no pair of orthogonal projectors are both assigned value $1$ (exclusivity), and that one of the projectors is assigned value $1$ in every complete basis set (completeness).

The KS theorem itself can be seen as a corollary of the remarkable Gleason's theorem \cite{Gleason57}. Simply put, this theorem states that when one assumes, in addition to the above non-contextuality, the mathematical structure of measurements in quantum theory, then one can recover the Born rule for calculating probabilities \cite{Piron76, Mackey63}. That is, in this case, \textit{any} assignment of outcome probabilities to projectors takes the form of applying the Born rule to some density operator. Gleason's work was not explicitly aimed at addressing the hidden variable problem, but was rather directed towards reducing the axiomatic basis of quantum theory. Nevertheless, the KS theorem that establishes the specific contradiction between the Born rule and non-contextual $\{0,1\}$-assignments can be seen as a consequence of Gleason's theorem along with a logical compactification argument \cite{Pitowsky98, Pitowsky06}. 

Pitowsky observed the above fact and suggested that it should be possible to find generalisations of the Kochen-Specker theorem to beyond $\{0,1\}$-colorings, noting that the "general constructive case is, to the best of my knowledge, open" \cite{Pitowsky98}. This paper is therefore propaedeutic to a larger project exploring such finite outcome alphabet generalisations of the Kochen-Specker theorem. Specifically here, we consider an interesting case going beyond $\{0,1\}$, namely non-contextual outcome assignments from $\mathcal{O} = \{0, p, 1-p, 1\}$ for $p \neq 1/d$ and $0 \leq p \leq 1/2$. Observe that the $d$-dimensional maximally mixed state naturally provides an outcome assignment of $1/d$ to all rank-one projective measurements. It is noteworthy that the known KS sets admit assignment from the set $\mathcal{O}$ (see Appendix A in \cite{Sup}). It is also noteworthy that as we show, the case $p = 1/2$ is of especially important physical significance, being the class of fundamentally binary outcome assignments, a natural general probabilistic analog of quantum theory studied in \cite{KVC17} (see Appendix B in \cite{Sup}). These are hidden variable models which postulate that at a fundamental level the measurements are binary (at most one of two measurement outcomes can occur for each measurement), in contrast to the single deterministic outcome assignment ruled out by KS. 

In this paper, we establish by an explicit construction, a generalisation of the Kochen-Specker theorem to rule out non-contextual outcome assignments from $\mathcal{O}$ to quantum projectors in three dimensions, i.e., $\mathcal{O} = \{0, p, 1-p, 1\}$ for $p \leq 1/2$, and $p \neq 1/3$. Formally, we construct a finite set of vectors (alternatively, their corresponding rank-one projectors) that do not admit any outcome assignment from the set $\mathcal{O}$ obeying the Kochen-Specker rules that no pair of orthogonal projections have sum of value assignments exceeding $1$, and that the sum of value assignments to the projectors in every complete measurement is equal to $1$. As stated above, the traditional KS theorem has found applications in several areas of quantum information. Here, we present an application of our generalised KS theorem to the construction of two-player pseudo-telepathy games that admit a winnning quantum strategy but cannot be won by players sharing classical resources even if these are augmented with the quintessential resource of no-signaling boxes of the type proposed by Popescu and Rohrlich in \cite{PR94} (that is, having marginal probabilities belonging to $\{0,1/2,1\}$).

The paper is organised as follows. We first introduce some preliminary notation and elementary concepts, in particular the representation of KS sets as graphs, and the notion of $\{0,1\}$-gadgets that were shown to be crucial in the construction of KS sets in \cite{RRHP+20, LRHRH23}. We then state the central theorem of the paper establishing a generalisation of the KS theorem to $\mathcal{O}$-valued outcome assignments, and outline a sketch of its proof, deferring the technical details to Appendix C in the Supplemental Material \cite{Sup}. We then present the construction of a class of two-player Pseudo-Telepathy games using the central theorem, detailing the uniqueness of this new class of games and highlighting their application in developing device-independent protocols for randomness generation against no-signaling adversaries. We conclude with a discussion on further possible generalisations of the Kochen-Specker theorem and possible new applications.  \\

\textit{Kochen-Specker and Gleason theorems.-} A projective measurement $M$ is described by a set $M = \{V_1, \ldots, V_d\}$ of projectors $V_i$ in a complex Hilbert space, that are orthogonal $V_i V_j = \delta_{i,j} V_i$ and sum to the identity $\sum_i V_i = \mathds{1}$. Each $V_i$ corresponds to a possible outcome $i$ of measurement $M$ and determines the probability of this outcome when measuring a state $| \psi \rangle$ through the Born rule $\text{Pr}_{\psi}(i | M) = \langle \psi | V_i | \psi \rangle$. When two physically distinct measurements $M = \{V_1, \ldots, V_d\}$ and $M' = \{V'_1, \ldots, V'_d\}$ share a common projector $V_i = V'_i = V$ it holds that $\text{Pr}_{\psi}(i | M) =\text{Pr}_{\psi}(i' | M') = \langle \psi | V | \psi \rangle$. That is, the outcome probabilities in quantum theory are are determined by the individual projectors alone, independently of the context to which they belong - the probability assignment is non-contextual. 

To prove the Kochen-Specker theorem \cite{KS67}, one usually considers a finite set of rank-one projectors in a complex Hilbert space of dimension $d \geq 3$. We represent the projectors by the vectors onto which they project and consider a set $\mathcal{S} = \big\{ | v_1 \rangle, \ldots, |v_n \rangle  \big\} \subset \mathbb{C}^d$ for $d \geq 3$. Consider any assignment $f: \mathcal{S} \rightarrow \{0,1\}$ that associates to each $|v_i\rangle$ a probability $f(|v_i \rangle) \in \{0,1\}$. To interpret the $f(|v_i \rangle)$ as outcome probabilities, they should satisfy the following two conditions termed as Kochen-Specker (KS) rules
\begin{enumerate}
\item Exclusivity: $\sum_{|v\rangle \in \mathcal{O}} f(|v\rangle) \leq 1$ for every set $\mathcal{O} \subseteq \mathcal{S}$ of mutually orthogonal vectors, 

\item Completeness: $\sum_{|v\rangle \in \mathcal{B}} f(|v\rangle) = 1$ for every set $\mathcal{B} \subseteq \mathcal{S}$ of $d$ mutually orthogonal vectors (a basis or complete measurement).

\end{enumerate} 
An assignment $f: \mathcal{S} \rightarrow \{0,1\}$ satisfying the KS rules is called a $\{0,1\}$-coloring. The KS theorem states that for $d \geq 3$, there exist finite sets of vectors (called KS sets) that do not admit a $\{0,1\}$-coloring, thus establishing the impossibility of a non-contextual deterministic outcome assignment. In their original proof, Kochen and Specker described a set $\mathcal{S}$ of $117$ vectors in $\mathbb{C}^d$ dimension $d = 3$ \cite{KS67}. Finding KS sets in arbitrary dimensions is a well-known hard problem towards which a huge amount of effort has been expended, in particular 'records' of minimal KS systems in different dimensions have been studied \cite{CG96, Arends09, AOW11, PMMM05}. The minimal KS set contains $18$ vectors in dimension $d = 4$ due to Cabello et al. \cite{Cabello08, CEG96}.   

The classical Gleason's theorem \cite{Gleason57, RB99, CKM85, JvN55} implies that in separable Hilbert spaces of dimension $d \geq 3$, a quantum state is completely determined by only knowing the answers to \textit{all} of the possible yes/no questions. Consider the real three-dimensional space $\mathbb{R}^3$, and suppose that $f$ is a non-negative function (technically termed a frame function) on the unit sphere with the property that $f(|u_1 \rangle) + f(|u_2 \rangle) + f(|u_3\rangle)$ is a fixed constant for \textit{every} set $\big\{|u_1 \rangle, |u_2 \rangle, |u_3 \rangle \big\}$ of mutually orthogonal vectors (a basis). Then, Gleason's theorem assures that $f$ is a quadratic form, i.e., there is a density operator $\rho$ such that $f(|u\rangle) = \langle u | \rho | u \rangle$ for every unit vector $|u \rangle \in \mathbb{R}^3$. The KS theorem then arises by a logical compactification argument from Gleason's theorem as a consequence of the fact that no density operator exists satisfying the above and an assignment $f(|u \rangle) = f(|v \rangle) = 1$ for two linearly independent vectors \cite{Pitowsky98} . 

Pitowsky provided a constructive method to prove a central lemma in the proof of Gleason's theorem, namely that every non-negative frame function on the set of vectors (strictly speaking, the rays) in $\mathbb{R}^3$ is continuous. He (along with Hrushovski) used this to prove a generalisation of the Kochen-Specker theorem in a different direction, that he termed the Logical Uncertainty Principle \cite{Pitowsky98, HP04} (for which we provided a simpler proof in \cite{RRHP+20}). Other attempts at constructive proofs of Gleason's theorem are also known \cite{RB99, CKM85, JvN55}. 

Closely related to the Logical Uncertainty Principle is the notion of a $01$-gadget \cite{RRHP+20}. This is a set of vectors $\mathcal{S}_{gad} \subset \mathbb{C}^d$ that is $\{0,1\}$-colorable, and contains two distinguished non-orthogonal vectors $|v_1 \rangle, |v_2 \rangle \in \mathcal{S}_{gad}$ for which $f(|v_1 \rangle) + f(|v_2 \rangle) \leq 1$ in every $\{0,1\}$-coloring $f$ of $\mathcal{S}_{gad}$. In other words, while a $01$-gadget $\mathcal{S}_{gad}$ admits a $\{0,1\}$-coloring, in any such coloring the two distinguished non-orthogonal vectors cannot both be assigned the value $1$ (as if they were actually orthogonal). The $01$-gadgets (also known as bugs \cite{Clifton93, Stairs83, Belinfante73} or true-implies-false sets \cite{CPSS18}) have been found to be instrumental in constructing proofs of the Kochen-Specker theorem, as well as more general sets exhibiting State-Independent Contextuality (SIC) \cite{YO12}. In this paper, we make use of $01$-gadgets to construct a different generalisation than Pitwosky's Logical Uncertainty Principle, namely to rule out outcome assignments from a different finite alphabet than $\{0,1\}$.

\textit{Generalisation of Kochen-Specker theorem.-}  
As noted above, Gleason's theorem implies that the only consistent (non-contextual) assignment rule, when one considers all the vectors in $\mathbb{R}^3$, is the Born rule. Pitowsky proved a finite version of Gleason's theorem called the Logical Uncertainty Principle which has been found to be of great use in contextuality-based randomness generation, where it is used in localizing the value indefiniteness implied by the KS theorem, i.e., in pinpointing the specific observables from whose outcomes randomness may be extracted \cite{ACS15}.

Gleason's theorem also suggests that it should be possible to find finite vector sets that do not admit non-contextual assignments (colorings) from other finite alphabets than $\{0,1\}$. As a first consideration, we might consider the alphabet $\mathcal{O} = \{0, p, 1-p, 1\}$ for $p \in [0,1/3) \cup (1/3, 1/2]$ (note that any rank-one projector set in $\mathbb{C}^d$ admits a coloring from the alphabet $\{0, 1/d, 1\}$ by virtue of the coloring resulting from measurements on the maximally mixed state $\frac{1}{d}\mathds{1}$). To be clear, let us formally define the notion of a $\mathcal{O}$-valued outcome assignment (coloring).

\begin{definition}
Let $\mathcal{O} := \{0, p, 1-p, 1\}$ with $p \in [0,1/3) \cup (1/3, 1/2]$. An $\mathcal{O}$-valued non-contextual outcome assignment of a set of vectors $\mathcal{S} = \{|v_1 \rangle, \ldots, |v_n \rangle\} \subset \mathbb{C}^d$ with $d \geq 3$ is an assignment $g: \mathcal{S} \rightarrow \mathcal{O}$ that satisfies the Kochen-Specker rules stated as:
\begin{enumerate}
\item $\sum_{|v\rangle \in \mathcal{B}'} g(|v\rangle) \leq 1$ for every set $\mathcal{B}' \subseteq \mathcal{S}$ of mutually orthogonal vectors, 

\item $\sum_{|v\rangle \in \mathcal{B}} g(|v\rangle) = 1$ for every set $\mathcal{B} \subseteq \mathcal{S}$ of $d$ mutually orthogonal vectors (a basis or complete measurement).

\end{enumerate} 
\end{definition}

From here on, for ease of exposition, we will specialise to the case $p = 1/2$. That is we will state the theorem for the case of $\{0,1/2,1\}$ colorings, the case of general  $\mathcal{O}$-colorings is explained in the Supp. Mat. \cite{Sup}. Every finite set of vectors that is not $\{0,1/2,1\}$-colorable provides a proof of the KS theorem, while not every proof set of the KS theorem is non-$\{0,1/2,1\}$-colorable. In fact, it turns out that all of the known KS vector sets (to the best of our knowledge) are $\{0,1/2,1\}$-colorable (see Appendix A of Supp. Mat \cite{Sup}). 

Therefore, a finite vector set that is not $\{0,1/2,1\}$-colorable proves a stronger statement than the original KS theorem, resulting in a \textit{stronger} notion of contextuality than hitherto explored, with concomitant novel applications as we shall see later. In this section, we state and prove by means of an explicit construction, the following theorem. As will be seen from the proof, some of the steps apply to more general finite alphabet outcome assignments than considered in this paper. 
\begin{thm}
\label{thm:main}
There exist finite vector sets in $\mathbb{C}^3$ that do not admit $\mathcal{O}$-valued outcome assignments. 
\end{thm}
\textit{Sketch of Proof}.-
The proof has three steps, each involving a gadget (a finite set of vectors) construction that proves an interesting lemma in its own right. In the first step, we present a gadget $\mathcal{S}_1$ that proves a lemma ruling out 'one-value' for any fixed projector in any $\mathcal{O}$-valued outcome assignment.
\begin{lemma}
\label{lem:one-value}
Let $|v_1 \rangle$ be any fixed vector in $\mathbb{C}^d$ with $d \geq 3$. There is a finite set of vectors $\mathcal{S}_1^{v} \subset \mathbb{C}^d$ with $|v \rangle \in \mathcal{S}^{v}$ such that there does not exist any outcome assignment $f: \mathcal{S}_1^{v} \rightarrow \mathcal{O}$ with $f(|v \rangle) = 1$.
\end{lemma}
The first step reduces the problem to the case of outcome assignments in $\mathcal{O} \setminus \{1\}$ since a finite vector set that does not admit outcome assignments from $\mathcal{O} \setminus \{1\}$ can be augmented with finite gadgets proving Lemma \ref{lem:one-value} for each of the vectors in the set. In the second step, we present a gadget $\mathcal{S}_2$  that admits outcome assignments from $\mathcal{O} \setminus \{1\}$ but such that in any such assignment it is necessarily the case that $d$ linearly independent vectors are all assigned value $0$. In the third and final step, we present a gadget that admits outcome assignments from $\mathcal{O} \setminus \{1\}$ but such that in any such assignment the $d$ linearly independent vectors of the form in step $2$ cannot all be assigned value $0$. Taking the union of the finite vector sets from steps $2$ and $3$, and augmenting them with a gadget from Lemma \ref{lem:one-value} for each of the vertices in $\mathcal{S}_2 \cup \mathcal{S}_3$ proves the statement. $\blacksquare$.

The above Theorem \ref{thm:main} restricted to the specific case of $\{0,1/2, 1\}$-assignments carries an additional physical significance arising from considerations of generalised probabilistic theories. Specifically, these correspond to the class of Fundamentally Binary theories - these are consistent (no-disturbance \cite{RSKK12}) theories in which measurements yielding many outcomes are constructed by selecting from binary measurements \cite{KVC17, KC16}.  In other words, these theories posit that on a fundamental level only measurements with two outcomes exist, and scenarios where a measurement has more than two outcomes are achieved by classical post-processing of one or more two-outcome measurements. 

In \cite{KVC17}, the authors showed that in considering Bell non-locality, the set of fundamentally binary non-signalling correlations does not encompass the set of quantum non-local correlations and hence an explanation in terms of measurements being binary is not feasible for quantum non-locality. Theories with binary measurements have an interesting property that was proven by us in \cite{LRHRH23}. Specifically in \cite{LRHRH23}, we characterized the extreme points of the set of correlations in fundamentally binary theories showing that for every measurement $x$ and every outcome $a$, it holds that $P_{A|X}(a|x) \in \{0,1/2,1\}$ for every extremal behaviour $\{P^{\text{ext}}_{A|X}(a|x)\}$ of the binary consistent correlation set. Seen in this light, a non-$\{0,1/2,1\}$-colorable vector set may also be stated as proving a Kochen-Specker-type theorem of the following form: for every quantum system belonging to a Hilbert space of dimension greater than two, irrespective of its actual state, a finite set of measurements exists that does not admit a consistent fundamentally binary outcome assignment (see Appendix B of Supp. Mat. \cite{Sup} for details). It is an interesting open question whether possible further generalisations to outcome assignments from different finite alphabet sets has a similar correspondence with other generalised probabilistic theories \cite{KC16}. 

Theorem \ref{thm:main} pertaining to the case of $\{0,1/2,1\}$ outcome assignments is also interesting from a graph-theoretic viewpoint, where it bears a relation with the stable set problem \cite{LSS87, CSW14}. A stable set (or independent set) $I$ in a graph is a subset of the vertex set such that no two vertices in $I$ are adjacent (connected by an edge). Determining the size of the maximum stable set (the independence number $\alpha(G)$) of a graph is a well-known hard problem, and can be seen as an optimization over the stable set polytope denoted by $STAB(G)$ (the convex hull of the incidence vectors of the stable sets of the graph).  Bounds on $\alpha(G)$ are obtained via a linear programming relaxation known as the fractional stable set polytope $FSTAB(G)$. In this light, Thm. \ref{thm:main} postulates the existence of graphs for which the intersection of $FSTAB(G)$ with the hyperplanes defined by the Completeness (Normalization) conditions for each maximum clique is empty (see Appendix B of \cite{Sup}).    


\textit{Novel PT games unwinnable using standard no-signalling resources.-} 
Pseudo-telepathy (PT) games \cite{BCT99} are distributed tasks that can be perfectly achieved with shared quantum - but not classical - information. Specifically, two distant parties who do not communicate but are allowed to share a certain entangled quantum state can satisfy a deterministic condition on their mutual input-output behavior with certainty, where parties without shared entanglement cannot do so. PT games are interesting from a foundational point of view as a simple 'all-versus-nothing' proof of quantum non-locality. In recent years, these games have also found applications as necessary resources in protocols of device-independent randomness amplification \cite{CR12, our16}, in demonstrating quantum computational advantage for shallow circuits \cite{BGK18} and in the proof of MIP* = RE \cite{JNVWY21}. 

An interesting connection between KS vector sets and PT games has been found \cite{RW04}. Specifically, every so-called 'weak' KS set leads to a PT game and every PT game in which the optimal quantum strategy uses a maximally entangled state leads to a KS set. Interestingly, the fact that known KS sets admit coloring in $\{0,1/2,1\}$ has meant that all known PT games are winnable when the classical players are given a specific type of additional no-signaling resource - viz. a non-local box with marginals in $\{0,1/2,1\}$. That is, all known PT games such as the Magic Square game used in the aforementioned applications are winnable when players share a single copy of such non-local boxes, the paradigmatic example of which is the PR-box \cite{PR94}. 

Here we present an analogous statement (proof in Appendix D of \cite{Sup}) showing that our generalised KS sets also lead to PT games, but now with the novel property that these games cannot be won even when the players share PR-type non-local boxes, where by PR-type behaviour we mean no-signalling boxes with marginals in $\{0, 1/2, 1\}$. 
\begin{proposition}
\label{prop:PT-games}
There exist bipartite Pseudo-telepathy games that cannot be won using classical resources even when supplemented with a single copy of PR-type behaviour in every game round.
\end{proposition}
The above Proposition \ref{prop:PT-games} readily leads to an application in device-independent cryptographic protocols for quantum random number generation or key distribution where the PT games give rise to protocols that are secure even against a class of no-signaling adversaries who are allowed to prepare the devices of honest players to obey the input-output behavior of these non-local boxes. Significant no-go theorems exist for proving security against general no-signalling adversaries \cite{AFTS12}, it is of interest to see whether such no-go statements can be overcome by restricting the power of no-signaling adversaries to an intermediate level above quantum. Other device-independent consequences of this novel class of Bell inequalities will be explored in forthcoming work.

\textit{Conclusions.-}
In this paper, we have studied a generalisation of the fundamental KS theorem to ruling out (consistent) hidden variable theories with outcome assignments from other finite alphabets than $\{0,1\}$. While we have presented one application to constructing bipartite Bell inequalities with novel device-independent applications, it would be of great interest to explore other applications to randomness generation, measurement-based quantum computation, and communication complexity scenarios where traditional KS contextuality has been shown to play a prominent role. The usual KS sets have thus far been seen to be necessary in constructing PT games and in identifying channels for which quantum entanglement enhances zero-error capacity \cite{CLMW10}. There is an interesting possibility that generalised KS sets could give rise to PT games for which the optimal strategy is achieved by a non-maximally entangled state, a result that would be useful in several areas. 

Simple and efficient criteria for identifying Kochen Specker orthogonality graphs have been found \cite{RH14}. It seems feasible to extend these to criteria for identifying orthogonality graphs of more general KS-type proofs, a task which we defer for future work. In traditional KS contextuality, records have been achieved of the minimal number of projectors needed to demonstrate state-independent and state-dependent \cite{KCBS08} contextuality in different dimensions. Such small sized proofs have enabled experimental tests of KS contextuality \cite{HLZPG03, LLS+11, W+22}. While our proof is constructive and involves a finite set of vectors, it would be important to reduce the number of projectors to make the generalised KS theorem amenable to experimental test. It would be thus very interesting to find the minimal sized KS sets ruling out $\mathcal{O}$-valued outcome assignments for different finite output alphabets $\mathcal{O}$. Similarly, various measures to quantify traditional KS contextuality have been proposed \cite{HZSS+23, GHH+14}, it would be of interest to quantify the stronger notion of contextuality captured by the generalised KS theorem.   




\textit{Acknowledgments.-} Useful discussions with Yuan Liu and Pawe{\l} Horodecki are acknowledged. We acknowledge support from the Early Career Scheme (ECS) Grant No. 27210620, the General Research Fund (GRF) Grant No. 17211122 and the Research Impact Fund (RIF) Grant No. R7035-21.

\bibliographystyle{apsrev4-2}
\bibliography{common}



\mbox{}


\appendix

\onecolumngrid
\section{Background on Kochen-Specker Theorem}
Much of the work on the KS theorem is done in terms of finite simple graphs, specifically by representing the orthogonality relations in a vector set $\mathcal{S}$ by a graph $G_{\mathcal{S}}$ known as the orthogonality graph of $\mathcal{S}$ \cite{LSS87}. In such a graph, each vector $|v_i \rangle \in \mathcal{S}$ is represented by a vertex $v_i$ of $G_{\mathcal{S}}$ and two vertices $v_1, v_2$ in the vertex set $V(G_{\mathcal{S}})$ are connected by an edge if the associated vectors $|v_1 \rangle, |v_2 \rangle$ are orthogonal, i.e., $v_1 \sim v_2$ if $\langle v_1 | v_2 \rangle = 0$.

A clique in the graph $G_{\mathcal{S}}$ is a subset of vertices $Q$ such that every pair of vertices in $Q$ is connected by an edge, i.e., $v_1 \sim v_2$ for all $v_1, v_2 \in Q$. It follows that in the orthogonality graph $G_{\mathcal{S}}$ represents a set of mutually orthogonal vectors in $\mathcal{S}$. If $\mathcal{S} \subset \mathbb{C}^d$ contains a basis set of $d$ mutually orthogonal vectors, then the maximum clique size of $G_{\mathcal{S}}$ is $d$, this is denoted as $\omega(G_{\mathcal{S}}) = d$ where $\omega(G)$ denotes the clique number of the graph.  

The notion of $\{0,1\}$-coloring for the vector set $\mathcal{S}$ can be translated into the problem of coloring the vertices of its orthogonality graph $G_{\mathcal{S}}$ such that vertices connected by an edge cannot both be assigned the color $1$ and every maximum clique has exactly one vertex of color $1$. We say that a graph $G$ is $\{0,1\}$-colorable if there exists an assignment $f: V(G) \rightarrow \{0,1\}$ such that
\begin{enumerate}
\item $\sum_{v \in Q} f(v) \leq 1$ for every clique $Q \subset V(G)$, 

\item $\sum_{v \in Q_{\max}} f(v) = 1$ for every maximum clique $Q_{\max} \subset V(G)$.
\end{enumerate} 
The KS theorem is then equivalent to the statement that there exist for any $d \geq 3$, finite vector sets $\mathcal{S} \subset \mathbb{C}^d$ such that their orthogonality graph $G_{\mathcal{S}}$ does not admit a $\{0,1\}$-coloring.

In a strictly analogous manner, one can consider the notion of outcome assignments from other finite output alphabets $\mathcal{O}$. Namely, we will say that a graph $G$ is $\mathcal{O}$-colorable if there exists an assignment $g: V(G) \rightarrow \mathcal{O}$ such that
\begin{enumerate}
\item $\sum_{v \in Q} g(v) \leq 1$ for every clique $Q \subset V(G)$, 

\item $\sum_{v \in Q_{\max}} g(v) = 1$ for every maximum clique $Q_{\max} \subset V(G)$.
\end{enumerate} 
Our strengthened KS theorem from the main text would then be equivalent to the statement that there exists for any $d \geq 3$ a finite vector set $\mathcal{S} \subset \mathbb{C}^d$ such that its orthogonality graph $G_{\mathcal{S}}$ does not admit a $\mathcal{O}$-coloring. 

In their original proof, Kochen and Specker described a set $\mathcal{S}$ of $117$ vectors in $\mathbb{C}^d$ dimension $d = 3$ \cite{KS67}. Finding KS sets in arbitrary dimensions is a well-known hard problem towards which a huge amount of effort has been expended, in particular 'records' of minimal KS systems in different dimensions have been studied \cite{CG96, Arends09, AOW11, PMMM05}. The minimal KS set contains $18$ vectors in dimension $d = 4$ due to Cabello et al. \cite{Cabello08, CEG96}.   

Intriguingly, it turns out that all of the known KS sets (to the best of our knowledge) admit a $\{0,1/2,1\}$-coloring. Some of the well-known KS sets are shown in Figs. \ref{fig:KS117}, \ref{fig:Peres} and \ref{fig:Cabello} together with a possible $\{0,1/2,1\}$-coloring. Even if perchance it were to turn out that a known KS set existed that did not admit any $\{0,1/2,1\}$-coloring, it would still be of great interest to develop novel techniques to find KS sets that do not admit outcome assignments from other finite alphabets. 
\begin{figure}\label{fig:coloring-KS}
\centering
\subfigure[A $\{0,1/2,1\}$ outcome assignment for the original $117$ vector KS set]{\label{fig:KS117}\includegraphics[width=.3\columnwidth]{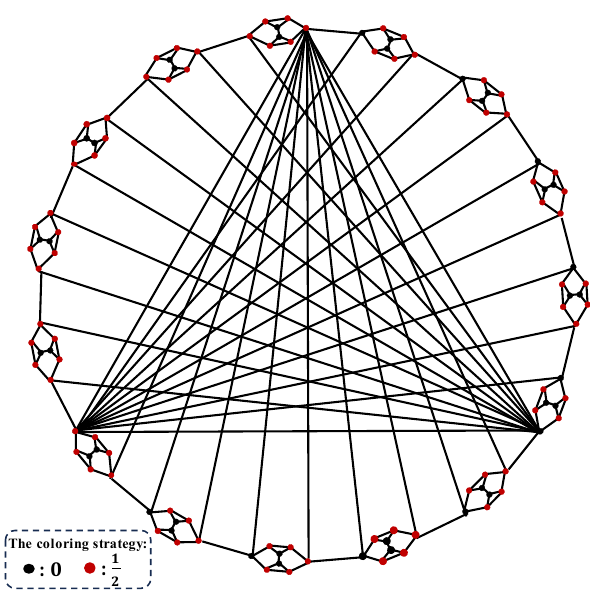}}\hspace{2em}
\subfigure[A $\{0,1/2,1\}$ outcome assignment for Peres' $33$ vector KS set]{\label{fig:Peres}\includegraphics[width=.3\columnwidth]{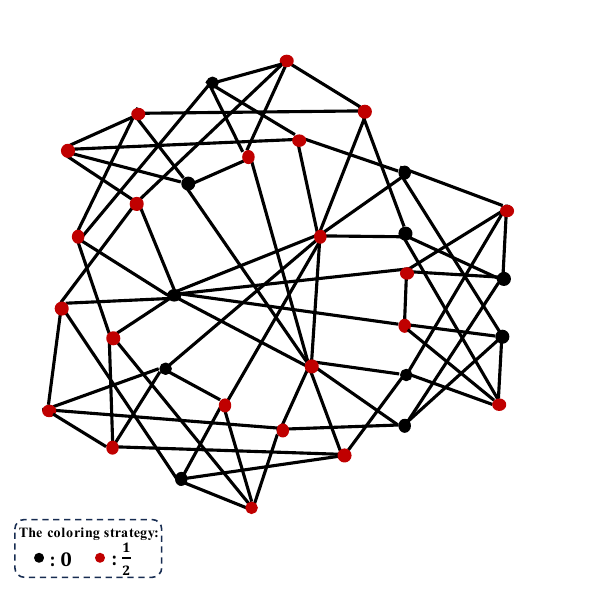}}\hspace{0.1em}\subfigure[A $\{0,1/2,1\}$ outcome assignment for Cabello's $18$ vector KS set ]
{\label{fig:Cabello}\includegraphics[width=.3\columnwidth]{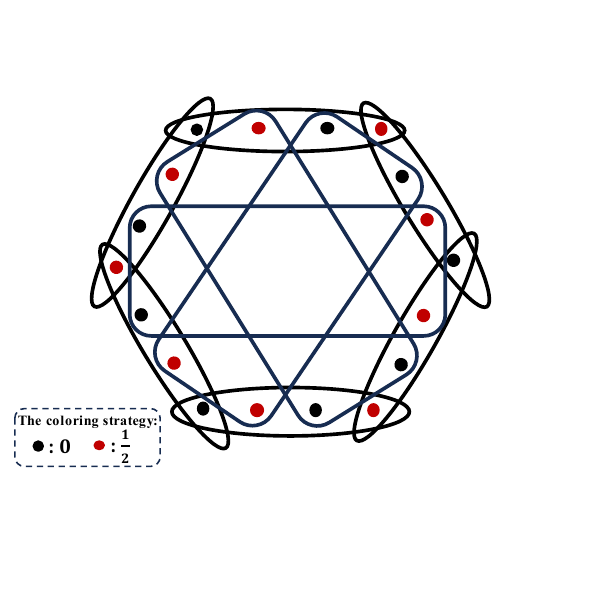}}\hspace{4em}
\caption{Some well-known Kochen-Specker sets are shown along with a possible outcome assignment from the set $\{0, 1/2, 1\}$ obeying the KS rules of exclusivity and completeness. (a) Kochen and Specker's original $117$ vector set \cite{KS67} in dimension three is shown. A valid $\{0, 1/2, 1\}$-coloring is indicated with red vertices being assigned value $1/2$ and black vertices being assigned value $0$. It can be readily checked that the sum of value assignments in every maximum clique (triangle) is $1$. (b) Peres' $33$ vector KS set in dimension three is shown. A valid $\{0, 1/2, 1\}$-coloring is indicated with red vertices being assigned value $1/2$ and black vertices being assigned value $0$. (c) Cabello's $18$ vector KS set \cite{Cabello08, CEG96} is shown as a hypergraph with hyperedges denoting maximum cliques of size four. A valid $\{0, 1/2, 1\}$-coloring is indicated with red vertices being assigned value $1/2$ and black vertices being assigned value $0$. It can be readily checked that the sum of value assignments in every maximum clique  is $1$.  }
\end{figure}


\section{Physical Interpretation of $\{0,1/2,1\}$-colorings as Fundamentally Binary Outcome Assignments}

In this section, we describe an intriguing interpretation of $\{0,1/2,1\}$-colorings in terms of a particular interesting class of generalised probabilistic theorems called "Fundamentally Binary Theories" \cite{KVC17}. This line of inquiry follows a huge research effort devoted to understanding physical and information-theoretic principles that enforce the quantum formalism. A natural class of alternatives to quantum theory are the Fundamentally Binary theories - these are no-signalling theories in which measurements yielding many outcomes are constructed by selecting from binary measurements. In other words, these theories posit that on a fundamental level only measurements with two outcomes exist, and scenarios where a measurement has more than two outcomes are achieved by classical post-processing of one or more two-outcome measurements. Fundamentally binary correlations are characterised as the convex hull of all consistent behaviours $\{P_{A|X}(a|x)\}$ obeying the constraint that for all measurements $x$, it holds that $P_{A|X}(a|x) = 0$ for all but (at most) two outcomes $a$.

Consider an orthogonality graph $G$ with a set of maximum cliques (contexts) $\mathcal{C} = \{C_1, \ldots, C_k\}$ where each clique $C_i$ if of size $\omega(G) = d$ (corresponding to a complete projective measurement). A behaviour (aka box or correlation) $\{P_{A|X}(a|x)\}$ is a set of conditional probability distributions with input $x \in \{C_1, \ldots, C_k\}$ and output $a \in \{1, \ldots, d\}$. A behaviour is said to be compatible with an orthogonality graph $G$ if it is a family of normalized probability distributions such that for each $c \in \{C_1, \ldots, C_k\}$, there is a corresponding probability distribution in this family. A behaviour is said to be consistent (or non-disturbing) if for all pairs $C_1, C_2$ it holds that
\begin{equation}
P_{A|X}(a=s|x = C_1) = P_{A|X}(a=s|x = C_2) \quad \forall s \in C_1 \cap C_2.
\end{equation}
In other words, the probability of the particular measurement outcome is independent of which context it is measured in. A consistent behaviour is said to be fundamentally binary if it belongs to the convex hull of consistent boxes  $\{P_{A|X}(a|x)\}$ for which for all measurements $x$, it holds that $P_{A|X}(a|x) = 0$ for all but (at most) two outcomes $a$, together with any box obtained by local classical post-processing of such boxes. That is, at most two outcomes in each measurement have non-zero probabilities (that sum to unity) and the remaining outcomes have probability $0$. 

In \cite{LRHRH23}, building on a graph-theoretic result for the so-called Fractional Stable Set Polytope $(FSTAB(G))$  for a graph $G$ by Nemhauser and Trotter \cite{NT75} and Balinski \cite{Balinski65}, we provided a characterization of the extremal points of the convex polytope of Fundamentally Binary correlations. 
\begin{pro}[\cite{LRHRH23}]
\label{prop:FB-ext}
Let $\{P^{\text{ext}}_{A|X}(a|x)\}$ be an extremal behaviour in the set of Fundamentally Binary correlations. Then for every measurement $x$ and every outcome $a$ it holds that $P^{\text{ext}}_{A|X}(a|x) \in \{0, 1/2, 1\}$.
\end{pro}
In light of Prop. \ref{prop:FB-ext}, our main theorem \ref{thm:main} can also be stated as a generalisation of the Kochen-Specker in the following way: for every quantum system belonging to a Hilbert space of dimension greater than two, irrespective of its actual state, a finite set of measurements exists that does not admit a fundamentally binary outcome assignment. The reason being that if a non-$\{0, 1/2, 1\}$-colorable vector set admits a fundamentally binary outcome assignment (even one that is not extremal), then the corresponding behaviour $\{P_{A|X}(a|x)\}$ is necessarily convex decomposable as $\{P_{A|X}(a|x)\} = \sum_i q_i P^{\text{ext}, (i),}_{A|X}(a|x)$ for some $q_i \geq 0, \sum_i q_i = 1$ and extremal behaviours $\{P^{\text{ext}, (i)}_{A|X}(a|x)\}$. However, as we have seen, for any extremal behaviour, the probabilities necessarily take values in $\{0, 1/2, 1\}$ which would provide a $\{0, 1/2, 1\}$-coloring for the vector set, giving a contradiction. Therefore, it follows that the set of Fundamentally Binary correlations for a non-$\{0, 1/2, 1\}$-colorable set is empty, giving rise to a Kochen-Specker-type theorem (recall that in the case of the non-$\{0,1\}$-colorable vector sets proving the original KS theorem, it is the polytope of classical behaviors obtained as the convex hull of deterministic $\{0,1\}$ behaviors that is empty).

\section{Proof of Theorem \ref{thm:main}.}
In this Appendix, we provide a proof of the main theorem, namely that there exist finite vector sets in $\mathbb{C}^3$ that do not admit any $\mathcal{O}$-valued outcome assignment, where $\mathcal{O} = \{0, p, 1-p, 1\}$ with $p \neq 1/3$. Our proof strategy is as follows. The proof consists of three steps. 

In the first step, we provide a construction to rule out any one-valued vectors in the $\mathcal{O}$-coloring. That is, we show that if for a finite vector set $\mathcal{S} = \{|v_1 \rangle, \ldots, |v_n \rangle\}$ and a coloring $f: \mathcal{S} \rightarrow \mathcal{O}$ it holds that $f(|v_1 \rangle) = 1$ then there exists set $\mathcal{S}_1^{v_1}$ such that augmenting $f$ to $\mathcal{S} \cup \mathcal{S}^{v_1}$ leads to a contradiction. In other words, the finite vector set $\mathcal{S} \cup \mathcal{S}_1^{v_1} \cup \mathcal{S}_1^{v_2} \cup \ldots \cup \mathcal{S}_1^{v_n}$ can only admit $\mathcal{O}$-valued outcome assignments $f$ for which $f(|v_i \rangle) \neq 1$ for all $|v_i \rangle \in \mathcal{S}$. We therefore reduce the problem to proving the existence of a finite vector set $\mathcal{S}$ that does not admit an outcome assignment in $\mathcal{O} \setminus \{1\}$. 

In the next steps, we go about constructing such a set $\mathcal{S}$. In the second step, we present a simple small vector set $\mathcal{S}_2$ (consisting of $9$ vectors in three dimensions) that has the property that in any $\mathcal{O} \setminus \{1\}$-outcome assignment $f$ it holds that a triple $(|v_1 \rangle, |v_2 \rangle, |v_3 \rangle)$ of linearly independent vectors (to be specific, one of twelve triples related by a symmetry transformation) are all assigned value $0$, i.e., $f(|v_1 \rangle) = f(|v_2 \rangle) = f(|v_3 \rangle) = 0$. In the third step, we present another gadget, a vector set $\mathcal{S}_3$ that includes the vectors $|v_1 \rangle, |v_2 \rangle, |v_3 \rangle$. The gadget $\mathcal{S}_3$ has the property that no outcome assignment $f: \mathcal{S}_3 \rightarrow \mathcal{O} \setminus \{1\}$ exists such that $f(|v_1 \rangle) = f(|v_2 \rangle) = f(|v_3 \rangle) = 0$. In other words, the conjunction of (the finite vector sets from) the three steps gives a vector set that does not admit $\mathcal{O}$-valued outcome assignments.

\subsection{Ruling out one-valued projectors in finite outcome assignments}
We first address the question of reducing the case of $\mathcal{O}$-valued assignments to $\mathcal{O} \setminus \{1\}$-assignments. That is, we show a construction that allows to rule out $f(|v \rangle) = 1$ for any fixed vector $|v\rangle$. While we will present this Lemma in the context of $\mathcal{O}$-valued assignments it will be apparent from the proof that the statement can be extended to arbitrary assignments from a finite set $\{p_1, p_2, \ldots, p_n, 1\}$. The proof of the Lemma builds upon \cite{RRHP+20} where we provided a simple proof of an Extended KS Theorem (Logical Indeterminacy Principle) first proposed by Pitowsky and Hrushovski in \cite{Pitowsky98, HP04}.  

\begin{theorem}[\cite{RRHP+20, Pitowsky98, HP04}]
\label{thm:LIP}
Let $|v_1 \rangle$ and $|v_2 \rangle$ be any two non-orthogonal vectors in $\mathbb{C}^d$ with $d \geq 3$. Then there is a finite set of vectors $S^{v_1, v_2}_{\text{LIP}} \subset \mathbb{C}^d$ with $|v_1 \rangle, |v_2 \rangle \in S^{v_1, v_2}_{\text{LIP}}$ such that for any $[0,1]$-assignment $f: S \rightarrow [0,1]$, it holds that if $f(|v_1 \rangle) = 1$ then $0 < f(|v_2 \rangle) < 1$.
\end{theorem}
\begin{figure}\label{fig:Lem1}
\centering
\subfigure[Gadget to rule out $\left(f(v_1) = 1 \right) \wedge \left( f(v_2) =1 \right)$]{\label{fig:LIPa}\includegraphics[width=.5\columnwidth]{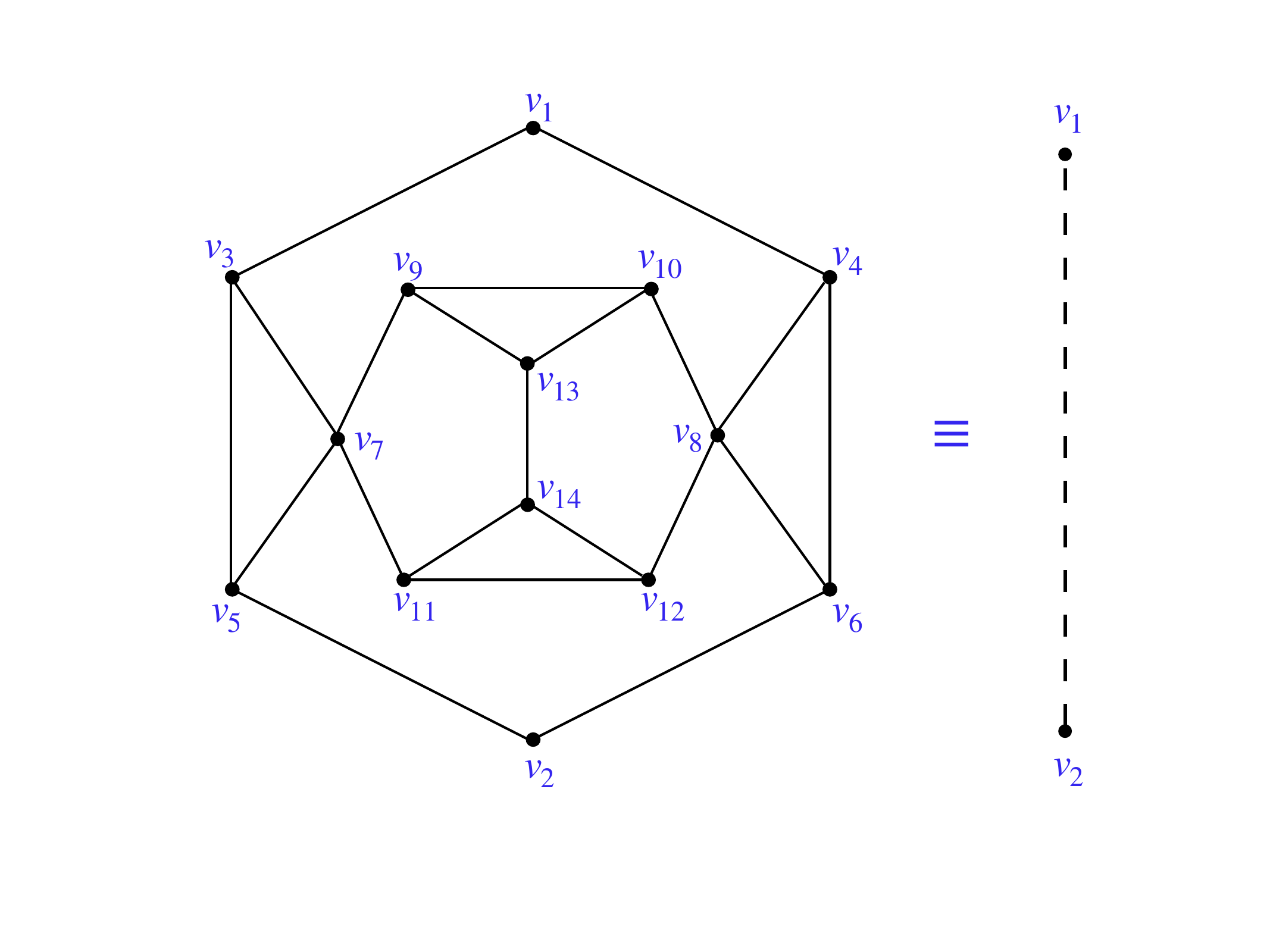}}
\subfigure[Gadget to show $\left(f(v_1) = 1 \right) \rightarrow \left(0 < f(v_2) < 1 \right)$ ]{\label{fig:LIPb}\includegraphics[width=.5\columnwidth]{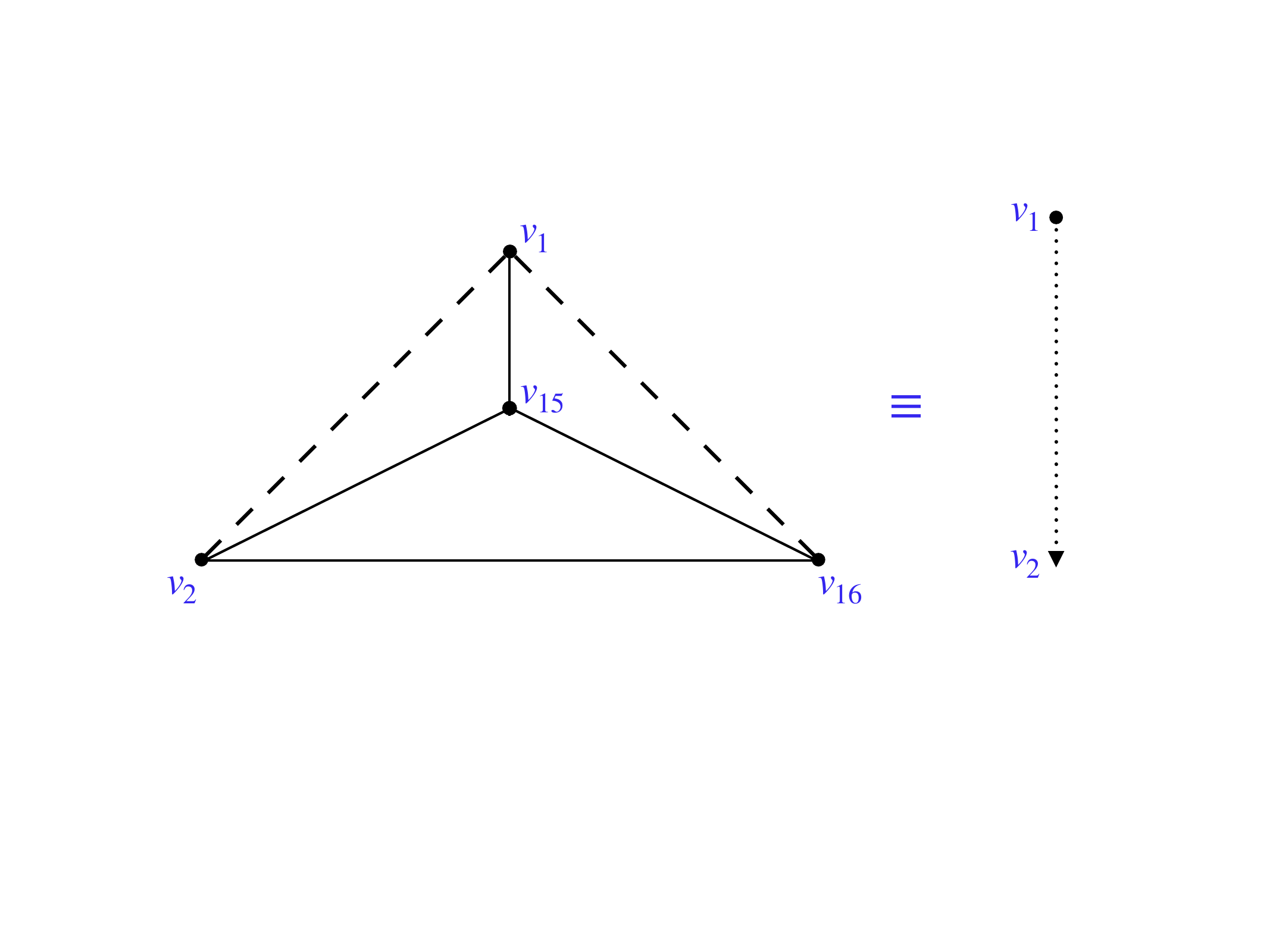}}\hspace{4em}
\caption{The gadgets to prove Thm. \ref{thm:LIP} are shown above. (a) Consider any assignment $f: S \rightarrow [0,1]$ for the set of (non-normalised represented by $\tilde{v}$) vectors $S$ defined as follows and represented by the orthogonality relations in the graph (a). $\langle \tilde{v}_1 |  =  (-\sqrt{2}, -1, 1)$, $\langle \tilde{v}_2 | = (\sqrt{2}, 1,1)$, $\langle \tilde{v}_3 | = (0,1,1)$, $\langle \tilde{v}_4| = (-2\sqrt{2}, 1, -3)$, $\langle \tilde{v}_5 | = (0,-1,1)$, $\langle \tilde{v}_6| = (2 \sqrt{2}, -1, -3)$, $\langle \tilde{v}_7 | = (1,0,0)$, $\langle \tilde{v}_8| = (1, 2\sqrt{2}, 0)$, $\langle \tilde{v}_9| = (0, \sqrt{3}, -1)$, $\langle \tilde{v}_{10} | = (-2 \sqrt{2}, 1, \sqrt{3})$, $\langle \tilde{v}_{11} | = (0, \sqrt{3}, 1)$, $\langle \tilde{v}_{12} | = (2 \sqrt{2}, -1, \sqrt{3})$, $\langle \tilde{v}_{13} | = (\sqrt{2}, 1, \sqrt{3})$, $\langle \tilde{v}_{14}| = (-\sqrt{2}, -1, \sqrt{3})$. It is readily checked that in any probabilistic assignment $f: S \rightarrow [0,1]$ it cannot be the case that both $f(|v_1 \rangle) = 1$ and $f(|v_2 \rangle) = 1$. While this construction is shown for the specific case
$\left|\langle v_2 |v_1 \rangle \right| = 1/2$ it readily extends to arbitrary $|v_1 \rangle, |v_2 \rangle \in \mathbb{C}^3$ through a repeated 'gadget-within-gadget' construction as shown in \cite{RRHP+20}. Gadgets of this type (for general overlap $\langle v_1 | v_2 \rangle$) forbidding assignments of value $1$ to $v_1, v_2$ are denoted by a dashed edge between $v_1$ and $v_2$. (b) Consider any assignment $f: S' \rightarrow [0,1]$ for the set of (non-normalised) vectors $S'$ defined as follows and represented by the orthogonality relations in the graph (b). $\langle \tilde{v}_1| = (-\sqrt{2}, -1, 1)$, $\langle \tilde{v}_2 | = (\sqrt{2}, 1, 1)$, $\langle \tilde{v}_{15} | = (-1, \sqrt{2}, 0)$, $\langle \tilde{v}_{16}| = (\sqrt{2},1,-3)$. In this figure, the dashed edges denote gadgets of the type in graph (a) i.e., such that $\left(f(v_1) =1 \right) \rightarrow \left(f(v_2) \neq 1 \right)$ and $\left(f(v_1) =1 \right) \rightarrow \left(f(v_{16}) \neq 1 \right)$ in any $[0,1]$-outcome assignment. Since $\left(f(v_1) = 1 \right) \rightarrow \left(f(v_{15}) = 0 \right)$ it follows that $\left(f(v_1) = 1 \right) \rightarrow \left(0 < f(v_2) < 1 \right)$ since $\left(f(v_2) = 0 \right)$ would imply $\left(f(v_{16}) = 1 \right)$ to satisfy the KS completeness condition. Gadgets of this type (having the property that $\left(f(v_1) = 1 \right) \rightarrow \left(0 < f(v_2) < 1 \right)$) are denoted by a dotted directed edge between $v_1$ and $v_2$.  }
\end{figure}

The explicit construction of $S^{v_1, v_2}_{\text{LIP}}$ proving Thm. \ref{thm:LIP} for arbitrary $|v_1 \rangle$ and $|v_2 \rangle$ was shown by us in \cite{RRHP+20} (in fact a slightly stronger statement was shown namely that $f(|v_1\rangle), f(|v_2\rangle) \in \{0,1\}$ if and only if $f(|v_1 \rangle) = f(|v_2 \rangle) = 0$). For completeness, we illustrate the central aspects of the proofs through Figs. \ref{fig:LIPa} and \ref{fig:LIPb}. The construction in Fig. \ref{fig:LIPa} shows a gadget with the property that in any $[0,1]$ assignment, one cannot have both end vertices taking value $1$, i.e., $f(v_1) = f(v_2) = 1$ is disallowed. While the construction is shown in Fig. \ref{fig:LIPa} for $|\langle v_2 | v_1 \rangle| = 1/2$, the construction for general $|v_1 \rangle, |v_2 \rangle$ follows a similar 'gadget-within-gadget' construction and can be seen in \cite{RRHP+20}.  The construction in Fig. \ref{fig:LIPb} shows a gadget with the property that in any $[0,1]$ assignment, if one end vertex takes value $1$ then the other cannot be deterministic, i.e., if $f(v_1) = 1$ then $0 < f(v_2) < 1$.
Note that the above Thm. \ref{thm:LIP} is concerned with general probabilistic assignments (from $[0,1]$ and obeying the KS rules) which contains the specific $\mathcal{O}$-valued assignments as a special case.

Theorem \ref{thm:LIP} states that for any two non-orthogonal vectors $|v_1 \rangle, |v_2 \rangle$ there is a finite set $S^{v_1, v_2}_{\text{LIP}}$ such that in all probabilistic assignments $f: S^{v_1, v_2}_{\text{LIP}} \rightarrow [0,1]$ it holds that if $f(|v_1 \rangle) = 1$ then $0 < f(|v_2 \rangle) < 1$. Therefore, if for any fixed vector $|v_1 \rangle$ it happens to be the case that $f(|v_1 \rangle) = 1$ then by adding the set $S^{v_1, v_2}_{\text{LIP}}$ for any other fixed vector $|v_2 \rangle$ we can ensure that $0< f(|v_2 \rangle) < 1$ for that probabilistic outcome assignment $f$. We will treat the construction of the set $S^{v_1, v_2}_{\text{LIP}}$ in Thm. \ref{thm:LIP} as a building block to prove the following Lemma. 
\begin{lemma}
\label{lem:LIP}
Let $|v_1 \rangle$ be any fixed vector in  $\mathbb{C}^d$ with $d \geq 3$. There is a finite set of vectors $\mathcal{S}_1^{v_1} \subset \mathbb{C}^d$ with $|v_1 \rangle \in \mathcal{S}_1^{v_1}$ such that there does not exist any outcome assignment $f: \mathcal{S}_1^{v_1} \rightarrow \mathcal{O}$ with $f(|v_1 \rangle) = 1$.
\end{lemma}
\begin{proof}

The proof goes through the orthogonality graph shown in the Fig. \ref{fig:Lem2}. In the graph, a directed dotted edge such as between $v_1$ and $v_3$ denotes the gadget constructed in Fig. \ref{fig:LIPb}. Unlike the usual edge in an orthogonality graph, this dotted edge denotes the fact that when the specific endpoint $v_1$ takes value $1$ then the other endpoint $v_3$ of the edge necessarily takes a value in $(0,1)$ (that is $0 < f(v_3) < 1$) for any $[0,1]$-assignment $f$. The corresponding orthogonal representation is given for the instance $|v_1 \rangle = (1,0,0)^T$. It follows that a corresponding orthogonal representation exists for any fixed vector in $\mathbb{C}^3$ through an appropriate unitary transformation. The vector set $\mathcal{S}_1^{v_1}$ is composed of the vectors $|v_1 \rangle,  |v_2 \rangle, \ldots, |v_6\rangle$ along with the vectors in the sets $S^{v_1, v_3}_{\text{LIP}}$, $S^{v_1, v_5}_{\text{LIP}}$ and $S^{v_1, v_6}_{\text{LIP}}$. 

\begin{figure}[h]
\label{fig:Lem2}
\includegraphics[width=.4\columnwidth]{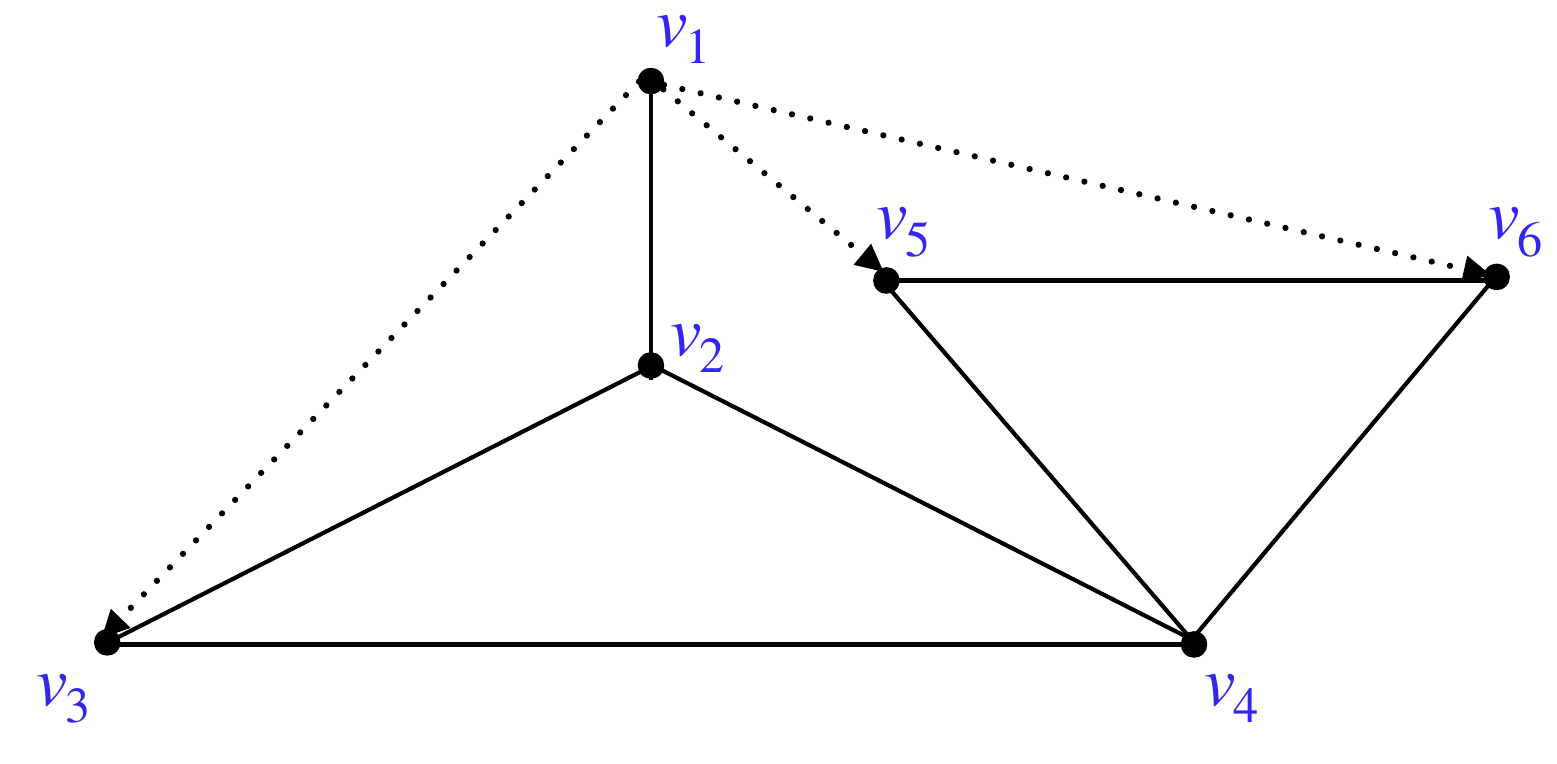}
\caption{The gadget to prove Lemma \ref{lem:LIP} is shown. In the figure, dotted directed edges denote the gadgets from Fig. \ref{fig:LIPb}. Consider the vector set $\mathcal{S}_1^{v_1}$ in the figure defined as $\langle v_1| = (1,0,0)$, $\langle v_2| = (0,0,1)$, $\langle v_3| = (1/\sqrt{2})(1,1,0)$, $\langle v_4| = (1/\sqrt{2})(-1,1,0)$, $\langle v_5| = (1/\sqrt{6})(1,1,2)$, $\langle v_6| = (1/\sqrt{3})(1,1,-1)$. Now consider any assignment $f: \mathcal{S}^{v_1} \rightarrow \{0,p,1-p,1\}$ with $p \leq 1/2$ and $p \neq 1/3$. Then if $f(v_1) = 1$ it holds that $0 < f(v_3) < 1$, $0 < f(v_5) < 1$ and $0 < f(v_6) < 1$ (by the property of the dotted edges) in addition to $f(v_2) = 0$ (by the KS exclusivity). The KS completeness condition gives $f(v_3) = f(v_5) + f(v_6)$ which in conjunction with the above gives $f(v_3) > f(v_5)$ and $f(v_3) > f(v_6)$. For $p \leq 1/2$ the above conditions can only be met when $f(v_3) = 1-p$, $f(v_5) = f(v_6) = p$ which however gives $p = 1/3$ which is excluded. Therefore no outcome assignment $f$ to the set $\mathcal{O}$ exists for the set $\mathcal{S}_1^{v_1}$ obeying $f(v_1) = 1$.   }
\end{figure}

Now suppose there exists an outcome assignment $f: \mathcal{S}_1^{v_1} \rightarrow \mathcal{O}$ with $f(|v_1 \rangle) = 1$. From the property stated in Lemma \ref{lem:LIP}, it follows that $0 < f(|v_3 \rangle) < 1$, $0 < f(|v_5 \rangle) < 1$ and $0< f(|v_6 \rangle) < 1$. And from the KS rules of exclusivity and completeness it follows that $f(|v_2 \rangle) = 0$, $f(|v_3 \rangle) + f(|v_4 \rangle) = 1$ and $f(|v_4 \rangle) + f(|v_5 \rangle) + f(|v_6 \rangle) = 1$. Therefore, we obtain that $f(|v_3 \rangle) = f(|v_5 \rangle) + f(|v_6 \rangle)$ with $0< f(|v_6 \rangle) < 1$ so that $f(|v_3 \rangle) > f(|v_5 \rangle)$. We also have  $f(|v_3 \rangle) = f(|v_5 \rangle) + f(|v_6 \rangle)$ with $0< f(|v_5 \rangle) < 1$ so that $f(|v_3 \rangle) > f(|v_6 \rangle)$. And since $f$ takes values in $\{0,p,1-p,1\}$ we have that for $p \leq 1/2$ the above conditions can only be met when $f(|v_3 \rangle) = 1-p, f(|v_5 \rangle) = f(|v_6 \rangle) = p$. However since $f(|v_3 \rangle) = f(|v_5 \rangle) + f(|v_6 \rangle)$ this gives $p = 1/3$ which is excluded. Therefore, no outcome assignment exists such that $f: \mathcal{S}_1^{v_1} \rightarrow \mathcal{O}$ and $f(v_1) = 1$ proving the statement. We note that when considering $\mathbb{C}^d$ for $d > 3$, the set $\mathcal{S}_1^v$ is simply modified as in traditional KS proofs by the addition of the computational basis vectors $(0,0,0,1,0,\ldots, 0)^T, \ldots, (0,0,0,0,\ldots, 0,1)^T$. In any assignment with $f(|v_1 \rangle) = 1$ these additional vectors are all assigned value $0$ so that the above argument directly extends.
\end{proof}

\subsection{Assigning $0$s to $d$ linearly independent vectors}
As explained, in the first step we have reduced the problem to proving the existence of a finite vector set $\mathcal{S}$ that does not admit an outcome assignment in $\mathcal{O} \setminus \{1\}$. Towards this end, in this subsection, we present a simple gadget that has the property that in any outcome assignment $f$ in the set $\mathcal{O} \setminus \{1\}$ it holds that a triple $(|v_1 \rangle, |v_2 \rangle, |v_3 \rangle)$ of linearly independent vectors (to be specific, one of twelve triples related by a symmetry transformation) are all assigned value $0$, i.e., $f(|v_1 \rangle) = f(|v_2 \rangle) = f(|v_3 \rangle) = 0$. Formally we consider the following set of vectors $\mathcal{S}_2$ in three dimensions represented by the orthogonality graph in Fig. \ref{fig:Lem3}.

\begin{figure}[h]
\label{fig:Lem3}
\includegraphics[width=.35\columnwidth]{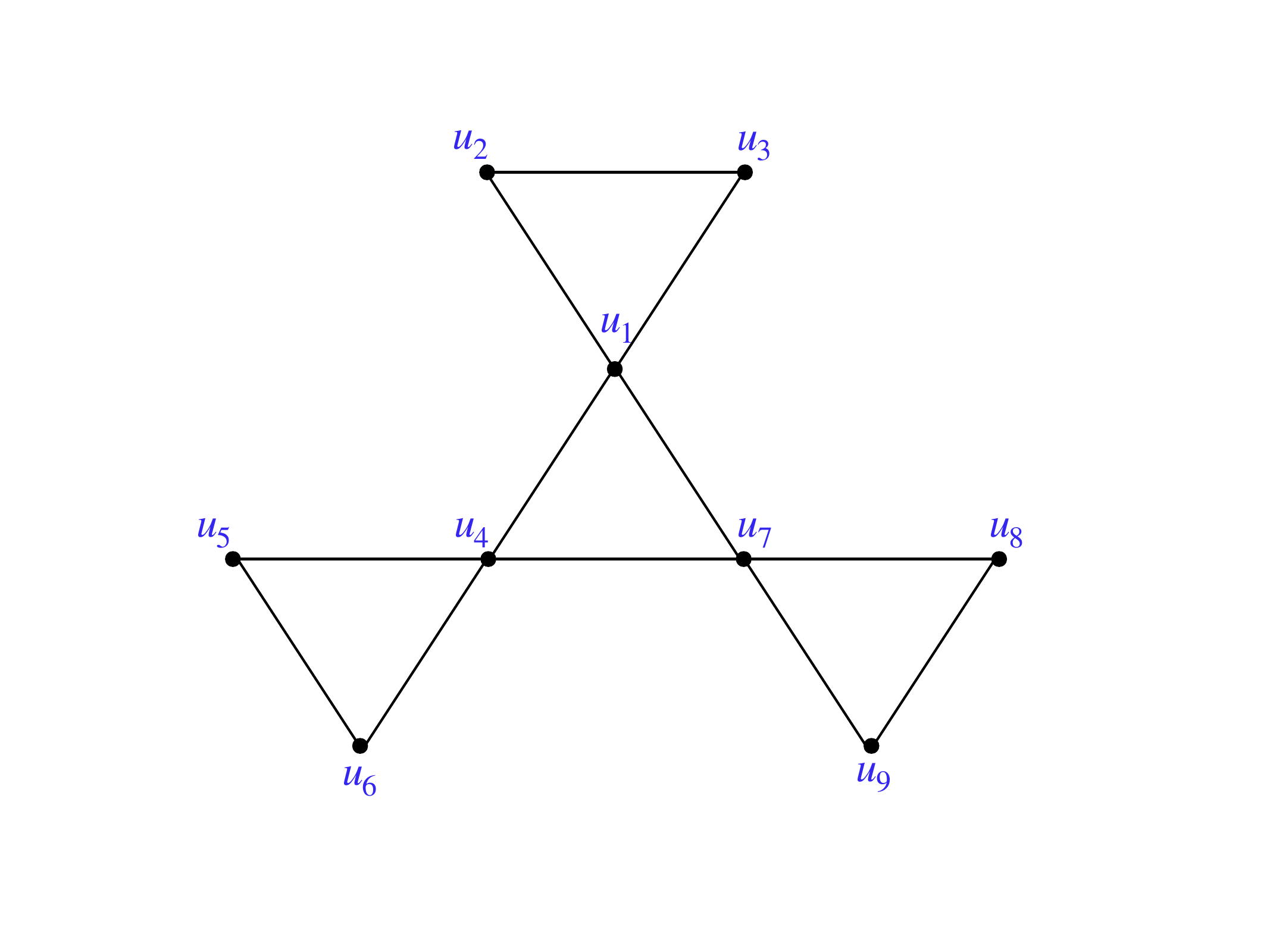}
\caption{In the figure is presented a simple gadget that achieves the property stated in Step 2 of the proof. Namely, consider the following vector set $\mathcal{S}_2$ represented by the orthogonality graph in the figure: $\langle u_1 | = (1,0,0)$, $\langle u_2 | = (0, \cos{\theta_1}, \sin{\theta_1})$, $\langle u_3 | = (0, -\sin{\theta_1}, \cos{\theta_1})$, $\langle u_4| = (0,1,0)$, $\langle u_5| =  (\cos{\theta_2},0, \sin{\theta_2})$, $\langle u_6 | = ( -\sin{\theta_2},0, \cos{\theta_2})$, $\langle u_7| = (0,0,1)$, $\langle u_8| =  (\cos{\theta_3}, \sin{\theta_3},0)$, $\langle u_9 | = (-\sin{\theta_3}, \cos{\theta_3},0)$. For simplicity we take $\theta_1 = \theta_2 = \theta_3 = \theta = \pi/3$. Consider any assignment $f: \mathcal{S}_2 \rightarrow \{0, p, 1-p\}$ with $p \leq 1/2$ and $p \neq 1/3$. To obey the KS completeness condition, it readily follows that we must have three linearly independent vectors in the set that are assigned value $0$. To be specific, any assignment $f$ must necessarily obey one of the following twelve possibilities: (i) $f(u_1) = f(u_5) = f(u_8) = 0$, (ii) $f(u_1) = f(u_5) = f(u_9) = 0$, (iii) $f(u_1) = f(u_6) = f(u_8) = 0$, (iv) $f(u_1) = f(u_6) = f(u_9) = 0$, (v) $f(u_4) = f(u_2) = f(u_9) = 0$, (vi) $f(u_4) = f(u_2) = f(u_8) = 0$, (vii) $f(u_4) = f(u_3) = f(u_9) = 0$, (viii) $f(u_4) = f(u_3) = f(u_8) = 0$, (ix) $f(u_7) = f(u_2) = f(u_6) = 0$, (x) $f(u_7) = f(u_2) = f(u_5) = 0$, (xi) $f(u_7) = f(u_3) = f(u_6) = 0$ or (xii) $f(u_7) = f(u_3) = f(u_5) = 0$. }
\end{figure}
\begin{eqnarray}
&&\langle u_1 | = (1,0,0), \;\; \langle u_2 | = (0, \cos{\theta_1}, \sin{\theta_1}), \;\; \langle u_3 | = (0, -\sin{\theta_1}, \cos{\theta_1}) \nonumber \\
&& \langle u_4| = (0,1,0), \;\; \langle u_5| =  (\cos{\theta_2},0, \sin{\theta_2}), \;\; \langle u_6 | = ( -\sin{\theta_2},0, \cos{\theta_2}) \nonumber \\
&& \langle u_7| = (0,0,1), \;\; \langle u_8| =  (\cos{\theta_3}, \sin{\theta_3},0), \;\; \langle u_9 | = (-\sin{\theta_3}, \cos{\theta_3},0). \nonumber \\ 
\end{eqnarray}
For simplicity we may take $\theta_1 = \theta_2 = \theta_3 = \theta = \pi/3$, say.
Consider any outcome assignment to this set $\mathcal{S}_2$ in $\mathcal{O} \setminus \{1\}$, i.e., $f: \mathcal{S}_2 \rightarrow \{0,p,1-p\}$ with $p \leq 1/2, p \neq 1/3$. Since each row as well as the first column corresponds to a basis in three dimensions, it follows that to comply with the KS rules, one vector from each row as well as the first column must necessarily be assigned value $0$ with the other vectors being assigned values $p, 1-p$. That is, it holds that in any assignment $f: \mathcal{S}_2 \rightarrow \{0,p, 1-p\}$, we must have $f(|u_1 \rangle) = 0 \wedge \left[f(|u_5 \rangle) = 0 \vee f(|u_6 \rangle) = 0\right] \wedge \left[f(|u_8 \rangle) = 0 \vee f(|u_9 \rangle) = 0 \right]$ or permutations thereof. To be precise, it holds that one of 12 linearly independent triples (listed below Fig. \ref{fig:Lem3}) has the property that all vectors in the triple are assigned value $0$ in any outcome assignment from the set $\mathcal{O} \setminus \{1\}$. \\ 



\subsection{Assignments that set $d$ linearly independent vectors to zero are identically zero} 
In this subsection, we build on Step 2 to construct the finite vector set that does not admit an outcome assignment in $\mathcal{O} \setminus \{1\}$. In conjunction with the Step 1 (by augmenting with finite vector sets ruling out value $1$ for each of the vertices), this set will prove the KS theorem for $\mathcal{O}$-valued outcome assignments.  

Our construction in this case is shown in Fig. \ref{fig:Step3}. Consider the set of vectors $\mathcal{S}_3$ in dimension three and their orthogonality graph shown in the figure and any outcome assignment $f: \mathcal{S}_3 \rightarrow \mathcal{O} \setminus \{1\}$ with the property that either (i) $f(w_1) = f(w_2) = f(w_3) = 0$ or (ii) $f(w_1) = f(w'_2) = f(w_3) = 0$. Here $| w_1 \rangle = (1,0,0)^T$, $| w_2 \rangle = (\cos{\theta}, \sin{\theta}, 0)^T$, $| w_3 \rangle = (\cos{\theta}, 0, \sin{\theta})^T$ and $| w'_2 \rangle = (-\sin{\theta}, \cos{\theta}, 0)^T$ with $\theta = \pi/3$. We recognize these vectors as the triples $(| u_1 \rangle, | u_8 \rangle, | u_5 \rangle)$ and $(| u_1 \rangle, | u_9 \rangle, | u_5 \rangle)$ from Step 2. 

Firstly, we note that an orthogonal representation in $\mathbb{R}^3$ exists for the graph in Fig. \ref{fig:Step3} (and consequently its induced subgraph in Fig. \ref{fig:Step3-2}). Namely, we fix $|w_1 \rangle = (1,0,0)^T$, $|w_2 \rangle = ((\cos{\theta}, \sin{\theta}, 0)^T$, $| w_3 \rangle = (\cos{\theta}, 0, \sin{\theta})^T$ and $| w'_2 \rangle = (-\sin{\theta}, \cos{\theta}, 0)^T$ with $\theta = \pi/3$. We first construct the set $\mathcal{S}^{(1)}_3$ consisting of the following set of vectors $|w_4 \rangle = (0, \cos{\phi_1}, \sin_{\phi_1})^T$, $|w_5 \rangle = (0, -\sin{\phi_1}, \cos{\phi_1})^T$,  $|w_6 \rangle = (0, \cos{\phi_2}, \sin{\phi_2})^T$, $|w_7 \rangle = (0, - \sin{\phi_2}, \cos{\phi_2})^T$, $|w_8 \rangle = |w_2 \rangle \times |w_4 \rangle$, $|w_9 \rangle = |w_2 \rangle \times |w_6 \rangle$, $|w_{10} \rangle = |w_3 \rangle \times |w_5 \rangle$, $|w_{11} \rangle = |w_3 \rangle \times |w_{7} \rangle$, $|w_{12} \rangle = |w_4 \rangle \times |w_8 \rangle$, $|w_{13} \rangle = |w_6 \rangle \times |w_9 \rangle$, $|w_{14} \rangle = |w_5 \rangle \times |w_{10} \rangle$, $|w_{15} \rangle = |w_7 \rangle \times |w_{11} \rangle$, $|w_{34} \rangle = |w_3 \rangle \times |w_{13} \rangle$. Here $\times$ denotes the usual vector cross product. We set $\theta = \pi/3$ and choose parameters $\phi_1, \phi_2$ to enforce the remaining orthogonality constraints: $\langle w_8 | w_9 \rangle = 0$, $\langle w_{10} | w_{11} \rangle = 0$ and $\langle w_{13} | w_{3} \rangle = 0$. The parameters $\phi_1, \phi_2$ are obtained analytically as the arctangents of the roots of algebraic equations of degree $4$, to avoid clutter we give their numerical values $\phi_1 \approx 5.7036$ and $\phi_2 \approx  3.5065$. We now proceed to show that no outcome assignment $f: \mathcal{S}_3 \rightarrow \mathcal{O} \setminus \{1\}$ exists, i.e., any such assignment leads to a contradiction with the KS rules of exclusivity and completeness. It is easy to see that in case (i) when we have $f(w_1) = f(w_2) = f(w_3) = 0$ then it is necessarily the case that $f(w_{12}) = f(w_{14}) = f(w_{13}) = f(w_{15}) = 0$ (these follow from the KS completeness rule). Now, since $f(w_{3}) = f(w_{13}) = 0$ it follows that completeness rule for the basis $(w_3, w_{13}, w_{34})$ cannot be satisfied, yielding a contradiction. Therefore in this case the contradiction is already reached from the subset $\mathcal{S}^{(1)}_3$ represented by the orthogonality graph in Fig. \ref{fig:Step3-2}. In our search over non-isomorphic graphs, this was the smallest graph that we found with the property required, namely that setting three linearly independent vectors to $0$ yields a contradiction with the KS rules of exclusivity and completeness for an outcome assignment in $\mathcal{O} \setminus \{1\}$.

We proceed to show the vectors representing the remaining vertices in the graph. Namely, we construct the set $\mathcal{S}^{(2)}_3$ consisting of the following set of vectors: $|w'_2 \rangle = (-\sin{\theta}, \cos{\theta},0)^T$, $| w_{16} \rangle= (0, \cos{\phi_3}, \sin{\phi_3})^T$, $| \tilde{w}_{17} \rangle = (0, -\sin{\phi_3}, \cos{\phi_3})^T$, $| w_{18} \rangle = (0, \cos{\phi_4}, \sin{\phi_4})^T$, $| w_{19} \rangle = (0, -\sin{\phi_4}, \cos{\phi_4})^T$, $|w_{20} \rangle = |w_{12} \rangle \times | w_{18} \rangle$, $|w_{21} \rangle = | w_{18} \rangle \times |w_{20} \rangle$, $|w_{22} \rangle = |w_{14} \rangle \times |w_{17} \rangle$, $|w_{23} \rangle = |w_{17} \rangle \times |w_{22} \rangle$, $|w_{24} \rangle = |w_{12} \rangle \times |w_{20} \rangle$, $|w_{25} \rangle  = |w_{14} \rangle \times |w_{22} \rangle$, $|w_{26} \rangle = |w_{16} \rangle \times |w_{24} \rangle$, $|w_{27} \rangle = |w_{19} \rangle \times |w_{25} \rangle$, $|w_{28} \rangle = (\cos{\theta_5} \cos{\phi_5}, \cos{\theta_5} \sin{\phi_5}, \sin{\theta_5})^T$, $|w_{29} \rangle = |w_{27} \rangle \times |w_{28} \rangle$, $|w_{30} \rangle = |w_{28} \rangle \times |w_{29} \rangle$, $|w_{31} \rangle = |w_{26} \rangle \times |w_{28} \rangle$, $|w_{32} \rangle = |w_{27} \rangle \times |w_{29} \rangle$, $|w_{33} \rangle = |w_9 \rangle \times |w_{30} \rangle$, and $|w_{35} \rangle = |w_2 \rangle \times |w'_2 \rangle$. Having fixed $\theta = \pi/3$, we choose the parameters $\phi_3, \phi_4, \phi_5, \theta_5$ to satisfy the remaining orthogonality constraints, namely $\langle w_{24} | w_{16} \rangle = 0$, $\langle w_{19} | w_{25} \rangle = 0$, $\langle w_{26} | w_{28} \rangle = 0$, $\langle w_9 |w_{30} \rangle = 0$. Again the angles $\phi_3, \phi_4, \phi_5, \theta_5$ are obtained analytically as the arctangents of the roots of algebraic equations of high degree, to avoid clutter we present their numerical values: $\theta_5 = 0$, $\phi_3 \approx 5.0234$, $\phi_4 \approx 0.5886$ and $\phi_5 \approx 2.1829$. An alternative approach to obtain a perhaps less unwieldy representation is to postulate a triple of real variables $(x_i, y_i, z_i)$ for each vertex $w_i$, express the constraints using polynomial equalities, namely $x_i^2 + y_i^2 + z_i^2 = 1$ and $x_i x_j + y_i y_j + z_i z_j = 0$ for every pair of adjacent vertices $w_i \sim w_j$, and solving the constraint system over the reals. 

Now, we proceed to show that no outcome assignment $f: \mathcal{S}_3 \rightarrow \mathcal{O} \setminus \{1\}$ exists obeying the case (ii) when we have $f(w_1) = f(w'_2) = f(w_3) = 0$. Again, applying the KS rules of exclusivity and completeness, we obtain in this case that $f(w_{14}) = f(w_{15}) = 0$ and either $f(w_8) = f(w_{13}) = 0$ or $f(w_9) = f(w_{12}) = 0$. When $f(w_{13}) = 0$ as we have already seen we reach a contradiction for the basis $(w_3, w_{13}, w_{34})$. In the case when $f(w_9) = f(w_{12}) = 0$, we obtain again by KS completeness that $f(w_{21}) = f(w_{26}) = f(w_{27}) = f(w_{23}) = 0$ which finally gives $f(w_{30}) = 0$. Now, since $f(w_{30}) = f(w_{9}) = 0$ it follows that the completeness rule for the basis $(w_{30}, w_{9}, w_{33})$ cannot be satisfied, yielding a contradiction.

The contradiction for the triples $(| u_1 \rangle, | u_9 \rangle, | u_6 \rangle)$ and $(| u_1 \rangle, | u_8 \rangle, | u_6 \rangle)$ can be achieved by a strictly analogous construction by considering the set $U^{u_1}_{\pi/2}|v \rangle$ for $|v\rangle \in \mathcal{S}_3$, i.e., by considering the vectors in the set $\mathcal{S}_3$ rotated by $\pi/2$ about the axis $| u_1 \rangle = (1,0,0)^T$. We denote this set as $U^{u_1}_{\pi/2} \mathcal{S}_3$.

The contradiction for the triples $(| u_4 \rangle, | u_9 \rangle, | u_2 \rangle)$, $(| u_4 \rangle, | u_8 \rangle, | u_2 \rangle)$, $(| u_4 \rangle, | u_9 \rangle, | u_3 \rangle)$ and $(| u_4 \rangle, | u_8 \rangle, | u_3 \rangle)$ can be achieved by an analogous construction considering the set $U^{u_7}_{\pi/2}|v \rangle$ for $|v \rangle \in \mathcal{S}_3 \bigcup U^{u_1}_{\pi/2} \mathcal{S}_3$, i.e., by considering the vectors in the set $\mathcal{S}_3 \bigcup U^{u_1}_{\pi/2} \mathcal{S}_3$ rotated by $\pi/2$ about the axis $| u_7 \rangle= (0,0,1)^T$. We denote this set as $U^{u_7}_{\pi/2} \left(\mathcal{S}_3 \bigcup U^{u_1}_{\pi/2} \mathcal{S}_3 \right)$. 

Finally, the contradiction for the triples $(| u_7 \rangle, | u_3 \rangle, | u_6 \rangle)$, $(| u_7 \rangle, | u_2 \rangle, | u_6 \rangle)$, $(| u_7 \rangle, | u_3 \rangle, | u_5 \rangle)$ and $(| u_7 \rangle, | u_2 \rangle, | u_5 \rangle)$ can be achieved by an analogous construction considering the set $U^{u_4}_{\pi/2}|v \rangle$ for $|v \rangle \in \mathcal{S}_3 \bigcup U^{u_1}_{\pi/2} \mathcal{S}_3$, i.e., by considering the vectors in the set $\mathcal{S}_3 \bigcup U^{u_1}_{\pi/2} \mathcal{S}_3$ rotated by $\pi/2$ about the axis $| u_4 \rangle = (0,1,0)^T$. We denote this set as $U^{u_4}_{\pi/2} \left(\mathcal{S}_3 \bigcup U^{u_1}_{\pi/2} \mathcal{S}_3 \right)$. The proof of Thm. \ref{thm:main} is completed by taking the union of these sets $\mathcal{S}_f :=\left(\mathcal{S}_3 \bigcup U^{u_1}_{\pi/2} \mathcal{S}_3 \right) \bigcup U^{u_7}_{\pi/2} \left(\mathcal{S}_3 \bigcup U^{u_1}_{\pi/2} \mathcal{S}_3 \right) \bigcup U^{u_4}_{\pi/2} \left(\mathcal{S}_3 \bigcup U^{u_1}_{\pi/2} \mathcal{S}_3 \right)$ along with the set $\mathcal{S}_1^v$ for each of the vectors $|v \rangle \in \mathcal{S}_f$.


\begin{figure}[h]
\subfigure[Gadget to prove Step 3]{\label{fig:Step3}\includegraphics[width=.6\columnwidth]{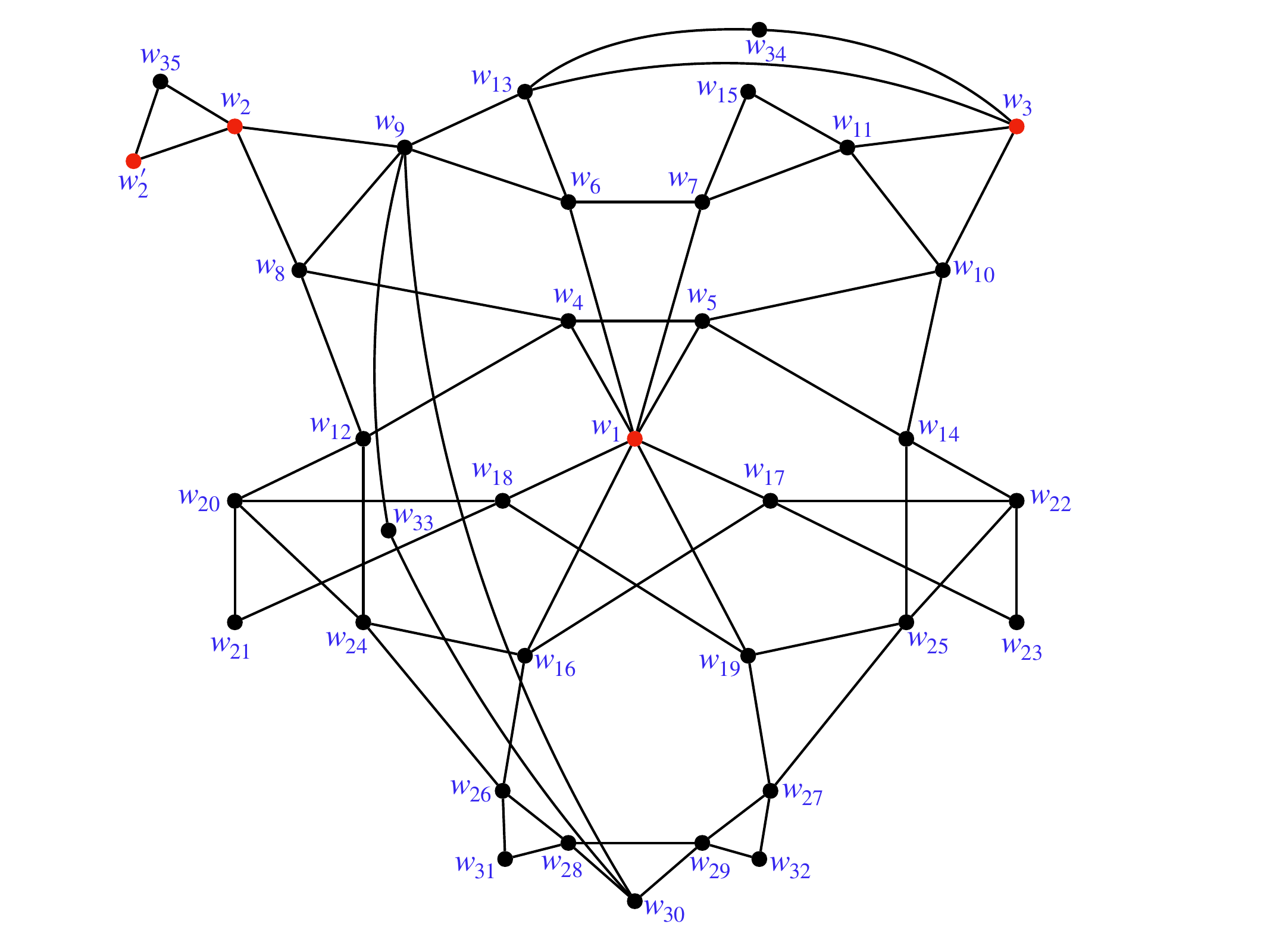}}
\subfigure[Smallest gadget for which a $\{0,1/2,1\}$ assignment leads to contradiction]{\label{fig:Step3-2}\includegraphics[width=.38\columnwidth]{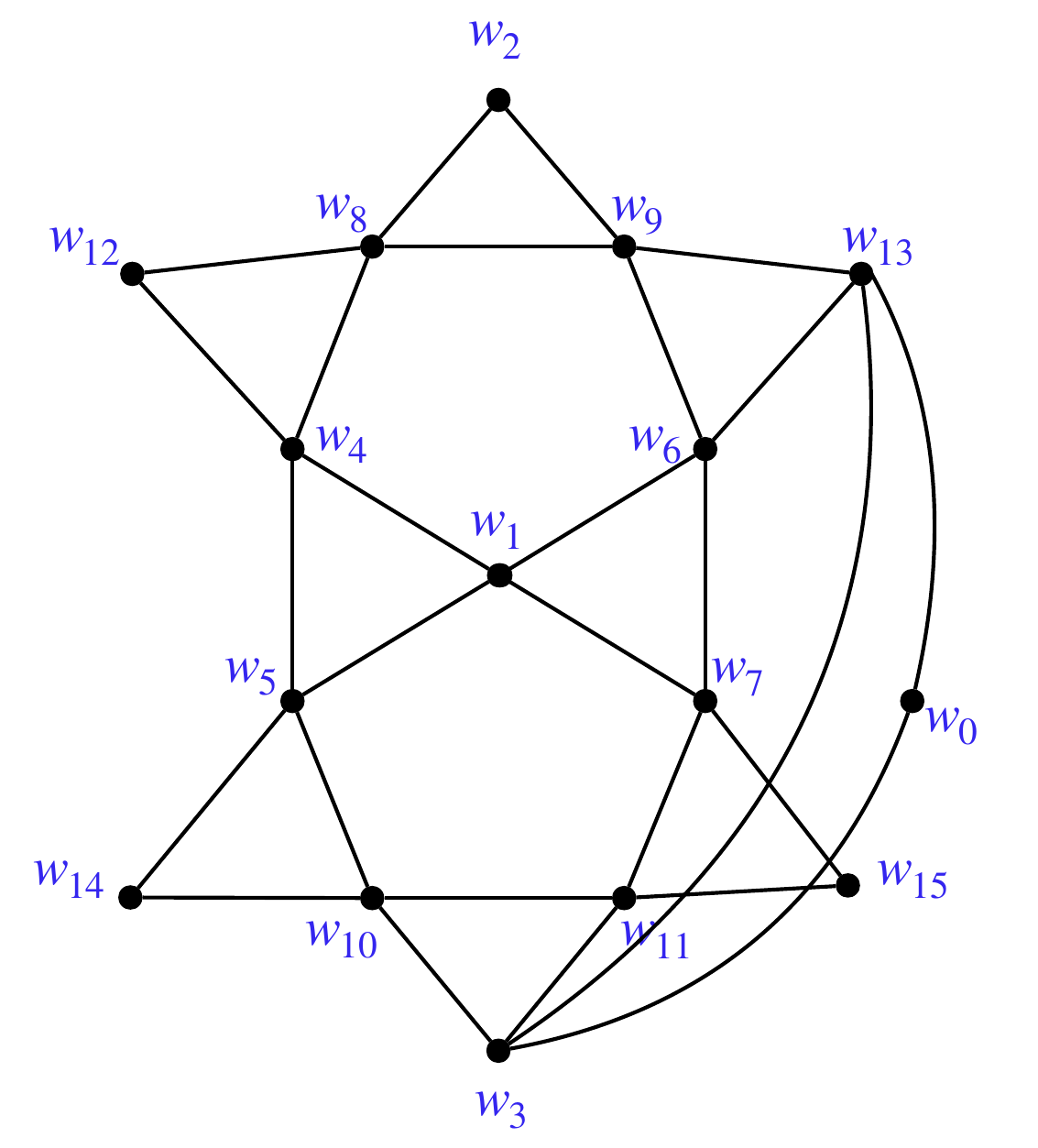}}\hspace{4em}
\caption{(a) In the figure on the left is presented a gadget that achieves the property stated in Step 3, namely that any assignment from $\mathcal{O} \setminus \{1\}$ that sets either of two triples $(|w_1 \rangle, |w_2 \rangle, |w_3 \rangle)$ or $(|w_1 \rangle, |w'_2 \rangle, |w_3 \rangle)$ to $0$ leads to a contradiction. This is seen through two cases: (i) when $f(w_1) = f(w_2) = f(w_3) = 0$ and (ii) when $f(w_1) = f(w'_2) = f(w_3) = 0$. Consider the following (non-normalised) vector set $\mathcal{S}^{(1)}_3 \cup \mathcal{S}^{(2)}_3$ represented by the orthogonality graph in the figure, where $\mathcal{S}^{(1)}_3$ is the set of the following vectors: $\langle w_1| = (1,0,0)$, $\langle w_2| = ( \cos{\theta}, \sin{\theta}, 0)$, $\langle w_3 | = (\cos{\theta}, 0, \sin{\theta})$, $\langle w_4| = (0, \cos{\phi_1}, \sin{\phi_1})$, $\langle w_5| = (0, -\sin{\phi_1}, \cos{\phi_1})$, $\langle w_6| = (0, \cos{\phi_2}, \sin{\phi_2})$, $\langle w_7|= (0, - \sin{\phi_2}, \cos{\phi_2})$, $\langle \tilde{w}_8| = (\sin{\theta} \sin{\phi_1}, -\cos{\theta} \sin{\phi_1}, \cos{\theta} \cos{\phi_1})$, $\langle \tilde{w}_9| = (\sin{\theta} \sin{\phi_2}, -\cos{\theta} \sin{\phi_2}, \cos{\theta} \cos{\phi_2})$, $\langle \tilde{w}_{10}| = (\sin{\theta} \sin{\phi_1}, -\cos{\theta} \cos{\phi_1}, -\cos{\theta} \sin{\phi_1})$, $\langle \tilde{w}_{11}| = (\sin{\theta} \sin{\phi_2}, -\cos{\theta} \cos{\phi_2}, -\cos{\theta} \sin{\phi_2})$, $\langle \tilde{w}_{12}| = (\cos{\theta}, \sin{\theta} \sin^2{\phi_1}, -\sin{\theta} \sin{\phi_1} \cos{\phi_1})$, $\langle \tilde{w}_{13}| = (\cos{\theta}, \sin{\theta} \sin^2{\phi_2}, -\sin{\theta} \sin{\phi_2} \cos{\phi_2})$, $\langle \tilde{w}_{14}| = (\cos{\theta}, \sin{\theta} \sin{\phi_1} \cos{\phi_1}, \sin{\theta} \sin^2{\phi_1})$, $\langle \tilde{w}_{15}| = (\cos{\theta}, \sin{\theta} \sin{\phi_2} \cos{\phi_2}, \sin{\theta} \sin^2{\phi_2})$, $|w_{34} \rangle = |w_{3} \rangle \times |w_{13} \rangle$. Now consider any assignment $f$ from the set $\{0, p, 1-p\}$ with $p \neq 1/3, p \leq 1/2$.  It is easy to see that in case (i) when we have $f(w_1) = f(w_2) = f(w_3) = 0$ then applying the KS completeness rule we obtain that $f(w_{12}) = f(w_{14}) = f(w_{13}) = f(w_{15}) = 0$. Now, since $f(w_{3}) = f(w_{13}) = 0$ it follows that completeness rule for the basis $(w_3, w_{13}, w_{34})$ cannot be satisfied, yielding a contradiction. Similarly, $\mathcal{S}^{(2)}_3$ is the set consisting of the following vectors: $|w'_2 \rangle = (-\sin{\theta}, \cos{\theta},0)^T$, $| w_{16} \rangle= (0, \cos{\phi_3}, \sin{\phi_3})^T$, $| \tilde{w}_{17} \rangle = (0, -\sin{\phi_3}, \cos{\phi_3})^T$, $| w_{18} \rangle = (0, \cos{\phi_4}, \sin{\phi_4})^T$, $| w_{19} \rangle = (0, -\sin{\phi_4}, \cos{\phi_4})^T$, $|w_{20} \rangle = |w_{12} \rangle \times | w_{18} \rangle$, $|w_{21} \rangle = | w_{18} \rangle \times |w_{20} \rangle$, $|w_{22} \rangle = |w_{14} \rangle \times |w_{17} \rangle$, $|w_{23} \rangle = |w_{17} \rangle \times |w_{22} \rangle$, $|w_{24} \rangle = |w_{12} \rangle \times |w_{20} \rangle$, $|w_{25} \rangle  = |w_{14} \rangle \times |w_{22} \rangle$, $|w_{26} \rangle = |w_{16} \rangle \times |w_{24} \rangle$, $|w_{27} \rangle = |w_{19} \rangle \times |w_{25} \rangle$, $|w_{28} \rangle = (\cos{\theta_5} \cos{\phi_5}, \cos{\theta_5} \sin{\phi_5}, \sin{\theta_5})^T$, $|w_{29} \rangle = |w_{27} \rangle \times |w_{28} \rangle$, $|w_{30} \rangle = |w_{28} \rangle \times |w_{29} \rangle$, $|w_{31} \rangle = |w_{26} \rangle \times |w_{28} \rangle$, $|w_{32} \rangle = |w_{27} \rangle \times |w_{29} \rangle$, $|w_{33} \rangle = |w_9 \rangle \times |w_{30} \rangle$, and $|w_{35} \rangle = |w_2 \rangle \times |w'_2 \rangle$. We choose the parameters $\theta = \pi/3$, $\phi_1, \phi_2, \phi_3, \phi_4, \phi_5, \theta_5$ to satisfy the remaining orthogonality constraints, namely $\langle w_{24} | w_{16} \rangle = 0$, $\langle w_{19} | w_{25} \rangle = 0$, $\langle w_{26} | w_{28} \rangle = 0$, $\langle w_9 |w_{30} \rangle = 0$, $\langle w_8 | w_9 \rangle = 0$, $\langle w_{10}| w_{11} \rangle = 0$, $\langle w_{13} | w_{3} \rangle = 0$.
Now in the case (ii) when we have $f(w_1) = f(w'_2) = f(w_3) = 0$, it is necessarily the case (to satisfy the KS conditions) that $f(w_{14}) = f(w_{15}) = 0$ and either $f(w_8) = f(w_{13}) = 0$ or $f(w_9) = f(w_{12}) = 0$. When $f(w_{13}) = 0$ as we have already seen we reach a contradiction for the basis $(w_3, w_{13}, w_{34})$. In the case when $f(w_9) = f(w_{12}) = 0$, we obtain again by KS completeness that $f(w_{21}) = f(w_{26}) = f(w_{27}) = f(w_{23}) = 0$ which finally gives $f(w_{30}) = 0$. Now, since $f(w_{30}) = f(w_{9}) = 0$ it follows that the completeness rule for the basis $(w_{30}, w_{9}, w_{33})$ cannot be satisfied, yielding a contradiction. (b) In the figure on the right is presented the smallest gadget (a subgraph of the graph on the left) with the property that any assignment $f: V(G) \rightarrow \{0,1/2,1\}$ satisfying $f(w_1) = f(w_2) = f(w_3) = 0$ leads to a contradiction, namely that the completeness rule for the basis $(w_3, w_{13}, w_0)$ cannot be satisfied. }
\end{figure}

\section{Applying Theorem \ref{thm:main} to constructing PT games unwinnable by PR-type boxes}

In this subsection, we will consider bipartite Bell inequalities or two-player non-local games that test for quantum non-locality rather than the single system quantum contextuality considered in the rest of the Appendix. Consider a bipartite Bell scenario in which two (non-communicating) players, Alice and Bob, choose measurement settings $x \in X$ and $y \in Y$ respectively, and obtain corresponding outcomes $a \in A$ and $b \in B$ respectively. Let $\{P_{A,B|X,Y}(a,b|x,y)\}$ denote the corresponding input-output behavior (the set of conditional probabilities of outputs given inputs) observed by the players. The winning probability $\omega(G)$ in a general two-player non-local game $G$ takes the form
\begin{equation}
\omega(G) := \sum_{\substack{x \in X\\, y \in Y}} \sum_{\substack{a \in A\\, b \in B}} \pi_{X,Y}(x,y) V(a,b,x,y) P_{A,B|X,Y}(a,b|x,y),
\end{equation}
where  $\pi_{X,Y}(x,y)$ denotes the distribution of inputs, $V(a,b,x,y) \in \{0,1\}$ is a predicate that denotes the winning condition (the condition that inputs and outputs should satisfy to win the game) and $\omega_c(G)$ denotes the maximum value achieved by classical (local hidden variable) strategies for the game $G$. A game $G$ is said to be a Pseudo-telepathy (PT) game if it allows for a perfect quantum strategy, i.e., if a quantum behavior $\{P^{(Q)}_{A,B|X,Y}(a,b|x,y)\}$ exists to allow Alice and Bob to win every round of $G$. Specifically, we have that $V(a,b,x,y) = 1$ for all $(a,b,x,y)$ such that $P^{(Q)}_{A,B|X,Y}(a,b|x,y) \neq 0$. Equivalently we have that the quantum winning probability is $\omega_q(G) = 1$ while no perfect classical strategy exists to allow the players to win every round, i.e., $\omega_c(G) < 1$. Quantum PT games are qualitatively different than statistical proofs of non-locality such as the CHSH game, and have been found to play a crucial role in the proofs of some fundamental results such as the quantum computational advantage for shallow circuits \cite{BGK18} and MIP* = RE \cite{JNVWY21}. Formally, we have
\begin{definition}[\cite{RW04}]
Let $| \psi \rangle \in \mathbb{H}_1 \otimes \mathbb{H}_2$ be a pure state. A Pseudo-Telepathy (PT) game with respect to $| \psi \rangle$ is a pair $(X, Y)$ where $X (Y)$ is a set of orthonormal bases of $\mathbb{H}_1 (\mathbb{H}_2)$ such that the following holds. Let $h$ be the function on $X \times Y$ defined as follows: $h(x, y)$ is the set of pairs $(|v_1 \rangle, |v_2 \rangle) \in x \times y$ for which $\langle \psi | v_1, v_2 \rangle \neq 0$. Then we have that for every classical strategy - pair of deterministic functions $(c_1, c_2)$ where $c_1 (c_2)$ is defined on $X (Y)$ and $c_1(x) \in x \; (\text{similarly} \; c_2(y) \in y)$ - there must exist specific bases $(x^*, y^*) \in X \times Y$ such that $(c_1(x^*), c_2(y^*)) \notin h(x^*, y^*)$. 
\end{definition}
In other words in the definition, we postulate the function $h$ on pairs of measurement bases which pinpoints for each input pair $x, y$ the set of outcomes $a = v_1, b = v_2$ that occur with non-zero probability in the optimal quantum strategy. As stated earlier, this set of outcomes defines the winning condition in the game, i.e., $V(a,b,x,y) = 1$ for all $(a,b,x,y)$ such that $P^{(Q)}_{A,B|X,Y}(a,b|x,y) \neq 0$. Then the game is PT if every classical strategy fails to satisfy the winning condition for at least one input pair $(x^*, y^*)$. That is, a classical strategy - which is defined by a pair of local deterministic assignments $c_1$ and $c_2$ that pick out a single measurement outcomes $c_1(x)$ and $c_2(y)$ for each input pair $(x, y)$ - necessarily achieves $P_{A,B|X,Y}(a,b|x^*, y^*) = 0$ for all $(a,b)$ for which $V(a,b,x^*, y^*) = 1$.

It is well-known since \cite{RW04} that every KS vector set in dimension $d$ can be used to construct a PT game for which the optimal winning strategy uses a maximally entangled state of local dimension $d$. However, known PT games such as the Magic Square game \cite{H+10} can be won when classical players additionally share a specific type of non-signalling behaviour known as the Popescu-Rohrlich box \cite{PR94}. The PR box is defined for $X = Y = A = B = \{0,1\}$ and is characterised by the property that it satisfies $a \oplus b = 1/2$ and has uniform marginals , i.e., $P_{A|X}(a|x) = P_{B|Y}(b|y) = 1/2$ for all $a, b, x, y$. 

In this subsection, we show that our generalized KS sets for $\{0, 1/2, 1\}$-colorings can be used to construct PT games with the novel property that they cannot be won even when classical players additionally have access to the resource of a (single copy of) PR-type behaviour in every game round. Specifically, we consider the sets $GKS$ that prove Thm. \ref{thm:main} from the main text with measurement bases completed, i.e., every set of mutually orthogonal vectors is completed to be part of a complete basis of $d$ mutually orthogonal vectors (a complete measurement) by augmenting the set $GKS$ with appropriate vectors. It follows that the addition of vectors preserves the property of the set of being non-$\{0,1/2,1\}$-colorable (since the underlying set already possesses this property).

Formally, we present the proof of Prop. \ref{prop:PT-games} from the main text which we formally elaborate as the following Proposition. 
\begin{proposition*}
\label{prop:PT-games-2}
Let $\mathbb{H}_1 = \mathbb{H}_2 = \mathbb{C}^3$ and let $|\phi \rangle$ denote the maximally entangled state in $\mathbb{H}_1 \otimes \mathbb{H}_2$, i.e., $|\phi \rangle := \frac{1}{\sqrt{3}} \left(|00 \rangle + |11 \rangle + |22 \rangle \right)$. Let $GKS \subset \mathbb{C}^3$ denote the generalized KS set from Theorem \ref{thm:main} for the case $\mathcal{O} = \{0, 1/2, 1\}$ with all measurement bases completed. Let us define
\begin{equation}
X = Y :=\big \{z \subset GKS \big| z \; \text{is an orthonormal basis of} \; \mathbb{C}^3 \big\}.
\end{equation} 
Then $(X, \bar{Y})$ is a PT game $G_{GKS}$ with respect to $|\phi \rangle$ with the property that $G_{GKS}$ cannot be won even when classical players additionally have access to the resource of a single copy of a PR-type behaviour in every game round.
\end{proposition*}
\begin{proof}
The proof follows straightforwardly from the standard construction of PT games using KS sets and from the property of the set $GKS$ that it does not admit a $\{0, 1/2, 1\}$-outcome assignment. 
We first observe that the set of classical behaviours is a polytope defined as the convex hull of deterministic behaviours, i.e. with $P_{A|X}(a|x) \in \{0, 1\}$ and $\{P_{B|Y}(b|y) \in \{0,1\}$. The set of strategies supplemented with PR-type behaviours is also a polytope defined as the convex hull of deterministic behaviours and the extremal non-local non-signalling behaviours with marginals $P_{A|X}(a|x), P_{B|Y}(b|y) \in \{0, 1/2, 1\}$.

Now suppose by contradiction that a perfect strategy exists for classical players supplemented with a non-local box with marginals in $\{0,1/2, 1\}$ for the game $G_{GKS}$. Since the winning probability of the game is linear in $P_{A|B|X,Y}(a,b|x,y)$ the optimal winning probability must be achievable by an extremal behaviour in this polytope, i.e., with  $P^{opt}_{A|X}(a|x), P^{opt}_{B|Y}(b|y) \in \{0, 1/2, 1\}$. The crucial observation is that such a strategy defines a $\{0, 1/2, 1\}$-coloring for the set $GKS$ measured by each player. 

To see this, define functions $f_A: GKS \rightarrow \{0, 1/2, 1\}$ by $f_A(|v\rangle) = P^{opt}_{A|X}(|v\rangle |x)$ and $f_A: GKS \rightarrow \{0, 1/2, 1\}$ by $f_B(|v\rangle) = P^{opt}_{B|Y}(|v\rangle |y)$. For the maximally entangled state $| \phi \rangle$ we see that
\begin{equation}
\langle \phi | v, \bar{v'} \rangle = \frac{1}{\sqrt{3}} \sum_{k} \langle k | v \rangle \langle k | \bar{v'} \rangle = \frac{1}{\sqrt{3}} \overline{\langle v | v' \rangle}. 
\end{equation}
So that $P^{(Q)}(a,b|x,y) = 0$ for every pair of measurements $x, y$ and orthogonal vectors $a = |v \rangle, b = |v' \rangle$. From the above, by considering $x = y = m$, we see that $P^{(Q)}(a,b|x = m,y =m) = 0$ whenever $a \neq b$ giving $f_A = f_B =: f$. Furthermore, we have that for every $x \in X$ it holds that $\sum_{|v\rangle \in x} f(|v \rangle) = 1$. Since every set of mutually orthogonal vectors has been augmented to form a complete basis, the above implies that $f : GKS \rightarrow \{0,1/2,1\}$ obeys the Kochen-Specker rules and forms a valid $\{0, 1/2, 1\}$-outcome assignment of $GKS$, which is a contradiction. Therefore, $(X, \bar{Y})$ is a PT game with respect to the maximally entangled state with the requisite property.  



\end{proof}


\end{document}